\documentclass[12pt]{article} 
\usepackage[colorlinks, citecolor={blue}]{hyperref}
\usepackage{url}
\usepackage{amsfonts,amscd,amssymb}
\usepackage{amsthm,amsmath,natbib} 
\usepackage{algorithmic,algorithm}
\usepackage{bm}
\usepackage{bbm} 
\usepackage{color}
\usepackage{verbatim}
\usepackage{graphicx}
\usepackage{subfigure}
\usepackage{caption}
\usepackage{setspace}
\usepackage{natbib}
\usepackage[margin=1in]{geometry}
\usepackage[textsize=tiny]{todonotes}
\usepackage{appendix}

\newtheorem{thm}{Theorem}
\newtheorem{lemma}{Lemma}

\newtheorem{mydef}{Definition}

\newtheorem{appxlem}{Lemma}[section]
\newtheorem{appxdef}{Definition}[section]

\newcommand{\argmax}{\operatornamewithlimits{argmax}}

\newcommand{\mb}{\mathbf}

\title{Birth/birth-death processes and their computable transition probabilities with biological applications}
\date{}
\author{
Lam Si Tung Ho \\
Department of Biostatistics \\ 
University of California, Los Angeles \\
\and
Jason Xu \\
Department of Biomathematics \\
University of California, Los Angeles \\
\and
Forrest W.~Crawford \\
Department of Biostatistics \\
Yale University \\
\and
Vladimir N.~Minin \\
Departments of Statistics and Biology \\
University of Washington \\
\and
Marc A.~Suchard \\
Departments of Biomathematics, Biostatistics and Human Genetics \\
University of California, Los Angeles
}

\begin{document}
\captionsetup[subfigure]{labelformat=empty}

\maketitle

\clearpage

\begin{abstract}
Birth-death processes track the size of a univariate population, but many biological systems involve interaction between populations, necessitating models for two or more populations simultaneously. 
A lack of efficient methods for evaluating finite-time transition probabilities of bivariate processes, however, has restricted statistical inference in these models.
Researchers rely on computationally expensive methods such as matrix exponentiation or Monte Carlo approximation, restricting likelihood-based inference to small systems, or indirect methods such as approximate Bayesian computation. 
In this paper, we introduce the \textit{birth/birth-death process}, a tractable bivariate extension of the birth-death process, where rates are allowed to be nonlinear.
We develop an efficient algorithm to calculate its transition probabilities 
using a 
continued fraction representation of their Laplace transforms.
Next, we identify several exemplary models arising in molecular epidemiology, macro-parasite evolution, and infectious disease modeling that fall within this class, and demonstrate advantages of our proposed method over existing approaches to inference in these models.
Notably, the ubiquitous stochastic susceptible-infectious-removed (SIR) model falls within this class, and we emphasize that computable transition probabilities newly enable direct inference of parameters in the SIR model.
We also propose a very fast method for approximating the transition probabilities under the SIR model via a novel branching process simplification, and compare it to the continued fraction representation method with application to the 17th century plague in Eyam.
Although the two methods produce similar maximum \textit{a posteriori} estimates, the branching process approximation fails to capture the correlation structure in the joint posterior distribution.


\end{abstract}

\clearpage

\section{Introduction}
Birth-death processes have been used extensively in many applications including evolutionary biology, ecology, population genetics, epidemiology, and queuing theory \citep[see e.g.][]{novozhilov2006, crawford2012, doss2013, rabier2014detecting, crawford2015sex}. However, establishing analytic and computationally practical formulae for their transition probabilities is usually difficult \citep{novozhilov2006}. 
The state-of-the-art method for computing the transition probabilities of birth-death processes proposed in \citet{crawford2012} enables statistical estimation 
for general birth-death processes 
using likelihood-based inference \citep{crawford2014JASA}. Unfortunately, birth-death processes inherently only track one population, and extending 
this technique
beyond the univariate case is nontrivial. Many applied models require the consideration of two or more interacting populations simultaneously to model behavior such as competition, predation, or infection. Examples of such bivariate models include epidemic models \citep{mckendrick1926applications,Kermack1927, griffiths1972}, predator-prey models \citep{hitchcock1986, owen2014}, genetic models \citep{rosenberg2003, Xu2015}, and within-host macro-parasite models \citep{drovandi2011}. 

The most general extensions of birth-death processes to bivariate processes are competition processes \citep{reuter1961}. These processes allow not only ``birth" and ``death" events in each population, but also ``transition" events where an individual moves from one population to the other. Unlike birth-death processes, few attempts have been made to compute the transition probabilities of competition processes or their special cases. Hence, researchers usually rely on classical continuous-time Markov chain methods such as matrix exponentiation and diffusion approximation. Unfortunately, these methods fail to leverage the specific structure of competition processes, and have several intrinsic limitations.  Matrix exponentiation methods compute the transition probability matrix ${\mathbf P}(t)$ by solving the matrix form of Kolmogorov's forward equation ${\mathbf P}'(t) = {\mathbf P}(t) {\mathbf Q}$ with initial condition ${\mathbf P}(0) = {\mathbf I}$, where ${\mathbf Q}$ is the instantaneous rate matrix of the process. While this equation admits a unique solution ${\mathbf P}(t) = \exp({\mathbf Q}t)$ \citep{ephraim2012bivariate}, numerical evaluation of the matrix exponential is often troublesome \citep{moler2003nineteen}.
Its computational cost via eigenvalue decomposition, for instance, is cubic in the size of the state-space and thus becomes computationally prohibitive even with moderately sized state-spaces \citep{drovandi2011, crawford2012}. For example, \citet{Keeling2008} demonstrate that computing transition probabilities via matrix exponentiation for the simplest epidemic models is practical only when modeling spread of an infectious disease through a very small population (e.g., 100 people). Moreover, matrix exponentiation can introduce serious rounding errors for certain rate matrices even for biologically reasonable values \citep{schranz2008pathological, crawford2012, crawford2014JASA}. Diffusion approximations, on the other hand, require the state-space to be large in order to justify approximating a discrete process by a continuous-valued diffusion process \citep{karev2005modeling, golightly2005bayesian}, and can often remain inaccurate for simulation even in settings with large state-spaces \citep{golightly2005bayesian}. Branching processes form another closely related class of processes, and have been used in a likelihood-based framework to study bivariate populations \citep{Xu2015}. Branching processes are at once more general than competition processes, permitting events that increment populations by more than one, and also more restrictive in that linearity is implied by an assumption that particles act independently. The latter assumption is limiting in epidemiological applications, for instance, which commonly feature non-linear interactions between populations.

The lack of a reliable method for computing transition probabilities in bivariate processes forces researchers to apply alternative likelihood-free approaches such as approximate Bayesian computation (ABC) \citep{blum2010hiv, drovandi2011, owen2014}. The ABC approach uses simulated and observed summary statistics to bypass likelihood evaluation. Nonetheless, this is not a panacea approach that can completely replace traditional likelihood-based methods. The ABC method itself has several sources for 
loss of information
such as non-zero tolerance, and non-sufficient summary statistics 
\citep{sunnaaker2013approximate}. 
The tolerance is an ad hoc threshold to decide whether ABC accepts a new proposal. 
If the tolerance is zero and the summary statistics are sufficient, ABC is guaranteed to return the correct posterior distribution.
In practice, however, tolerance is always positive which often leads to bias.
In the context of counting processes, sufficient summary statistics usually do not exist because the data are observed partially. 
Thus, credible interval estimates under ABC are potentially inflated due to the loss of information \citep{csillery2010approximate}. 
Also, when sufficient summary statistics are not available, the ABC method can not be trusted in selecting between models \citep{robert2011lack}. 
Because of all these limitations, direct likelihood-based methods are often more favorable. 

In this paper, we develop an efficient method to compute the transition probabilities of a subclass of competition processes with two interacting populations of particles, enabling likelihood-based inference. We call this subclass birth(death)/birth-death processes, whose first population is increasing (decreasing). 
It is worth mentioning that we do not impose linearity condition for the rates of these processes.
A rigorous characterization of this class of processes and derivation of recursive formulae to compute their transition probabilities are provided in Section \ref{sec:bbd}. Our main tools are the Laplace transform and continued fractions that have been successfully applied for univariate birth-death processes in \citet{crawford2012}. These formulae enable accurate and computationally efficient numerical computation of transition probabilities. We implement this method in the new \texttt{R} package \texttt{MultiBD} \url{https://github.com/msuchard/MultiBD}. In Section \ref{sec:app}, we discuss multiple scientifically relevant applications of birth(death)/birth-death processes including stochastic susceptible-infectious-removed (SIR) models in epidemiology \citep{mckendrick1926applications,Kermack1927, raggett1982}, monomolecular reaction systems \citep{jahnke2007solving}, a birth-death-shift model for transposable elements \citep{rosenberg2003, Xu2015}, and a within-host macro-parasite model \citep{riley2003, drovandi2011}. We examine the accuracy of our method in simulation studies, including comparisons to branching process, matrix exponentiation method, and Monte Carlo approximations. Finally, we apply our method to estimate infection rates and death rates during the plague of Eyam in 1666 within a likelihood-based Bayesian framework in Section \ref{sec:eyam}.

\paragraph{Previous work on computing the transition probabilities:}
Analytic expressions of the transition probabilities have only been found for some special cases such as linear birth-death processes \citep[see e.g.][]{novozhilov2006} and monomolecular reaction systems \citep{jahnke2007solving}.
Therefore, matrix exponentiation is still the most common method for computing the transition probabilities of general Markov processes.
The state-of-the-art software package for exponentiating sparse matrices is Expokit \citep{sidje1998,moler2003nineteen}, which uses Krylov subspace projection method.
\citet{van2006preconditioning} propose a modified version using a simple preconditioned transformation to improve the convergence behavior of this method.
Although matrix exponentiation has the advantage of generality in that it can be applied to any Markov process, it is not the most efficient method in many scenarios. 
Recently, \citet{crawford2012} propose an efficient method for evaluating the transition probabilities of general birth-death processes using Laplace transform and continued fraction.
However, efficient methods that extend this result to general bivariate birth-death processes have yet to be found.


\section{Birth(death)/birth-death processes}
\label{sec:bbd}


\subsection{Birth/birth-death processes}
A birth/birth-death process is a bivariate continuous-time Markov process ${\bf X}(t) = (X_1(t), X_2(t))$, $t \geq 0$, whose state-space is in $\mathbb{N} \times \mathbb{N}$, the Cartesian product of the non-negative integers. We can describe a birth/birth-death process as governing dynamics of a system consisting two types of particles, where one out of four possible events can happen in infinitesimal time: (1) a new {\bf type 1} particle enters the system; (2) a new {\bf type 2} particle enters the system; (3) a  {\bf type 2} particle leaves the system; or (4) a {\bf type 2} particle becomes a {\bf type 1} particle. In this system, $X_1(t)$ and $X_2(t)$ track the number of {\bf type 1}  and {\bf type 2} particles at time $t$ respectively. Mathematically, there are five possibilities for ${\bf X}(t)$ during a small time interval $(t,t+dt)$:
\begin{align}
& \Pr \left \{ \begin{array}{l | l} X_1(t+dt) = a+1 & X_1(t)= a \\ X_2(t+dt) = b & X_2(t) = b \end{array} \right \} = \lambda^{(1)}_{ab} dt + o(dt) \nonumber \\
& \Pr \left \{ \begin{array}{l | l} X_1(t+dt) = a & X_1(t)= a \\ X_2(t+dt) = b+1 & X_2(t) = b \end{array} \right \} = \lambda^{(2)}_{ab} dt + o(dt) \nonumber \\
& \Pr \left \{ \begin{array}{l | l} X_1(t+dt) = a & X_1(t)= a \\ X_2(t+dt) = b-1 & X_2(t) = b \end{array} \right \} = \mu^{(2)}_{ab} dt + o(dt) \nonumber \\
& \Pr \left \{ \begin{array}{l | l} X_1(t+dt) = a+1 & X_1(t)= a \\ X_2(t+dt) = b-1 & X_2(t) = b \end{array} \right \} = \gamma_{ab}dt + o(dt) \nonumber \\
& \Pr \left \{ \begin{array}{l | l} X_1(t+dt) = a~~~~~ & X_1(t)= a \\ X_2(t+dt) = b & X_2(t) = b \end{array} \right \} = 1 - (\lambda^{(1)}_{ab} +\lambda^{(2)}_{ab} + \mu^{(2)}_{ab} + \gamma_{ab})dt + o(dt),
\label{eqn:bbd}
\end{align}
where $a,b \in \mathbb{N}$, $\lambda^{(1)}_{ab} \geq 0$ is the birth rate of {\bf type 1} particles given $a$ {\bf type 1} particles and $b$ {\bf type 2} particles, $\lambda^{(2)}_{ab} \geq 0$ is the equivalent birth rate of {\bf type 2} particles, $\mu^{(2)}_{ab} \geq 0$ is the death rate of {\bf type 2} particles, and $\gamma_{ab}$ is the transition rate from {\bf type 2} particles to {\bf type 1} particles. 
We fix $\lambda^{(1)}_{-1,b} = \lambda^{(2)}_{a,-1} = \mu^{(2)}_{a0} = \gamma_{-1,b} = \gamma_{a0} = 0$.

Letting $P^{a_0b_0}_{ab}(t) = \Pr \{{\bf X}(t)= (a,b)~ |~{\bf X}(0)=(a_0,b_0)\}$, 
the forward Kolmogorov's equations for the birth/birth-death process are
\begin{align}
\frac{dP^{a_0b_0}_{ab}(t)}{dt} &= \lambda^{(1)}_{a-1,b} P^{a_0b_0}_{a-1,b}(t) +  \lambda^{(2)}_{a,b-1} P^{a_0b_0}_{a,b-1}(t) +  \mu^{(2)}_{a,b+1}P^{a_0b_0}_{a,b+1}(t) \nonumber \\ 
&+ \gamma_{a-1,b+1}P^{a_0b_0}_{a-1,b+1}(t) - (\lambda^{(1)}_{ab} + \lambda^{(2)}_{ab} + \mu^{(2)}_{ab} + \gamma_{ab})P^{a_0b_0}_{ab}(t),
\label{eqn:trans_eq}
\end{align}
for all $(a,b)$.
In practice, we can usually only observe the process discretely.
In this scenario, the likelihood function is the product of transition probabilities between consecutive observations.
Therefore, computing $P^{a_0b_0}_{ab}(t)$ is an important step for any direct likelihood-based analysis.

In general, a birth/birth-death process is a special case of a competition process \citep{reuter1961} with rate matrix ${\bf Q} = \{ q_{ij} \}$ where $i,j \in \mathbb{N} \times \mathbb{N}$ and
\begin{center}
  \begin{tabular}{  c|c|c }
    $j$ & Competition process & Birth/birth-death\\	
    \hline
    $(a+1,b)$ & $q_{(a,b)(a+1,b)~~}$ & $\lambda^{(1)}_{ab}$\\
    $(a-1,b)$ & $q_{(a,b)(a-1,b)}~~$ & $0$\\
    $(a,b+1)$ & $q_{(a,b)(a,b+1)}~~$ & $\lambda^{(2)}_{ab}$\\
    $(a,b-1)$ & $q_{(a,b)(a,b-1)}~~$ & $\mu^{(2)}_{ab}$\\
    $(a+1,b-1)$ & $q_{(a,b)(a+1,b-1)}$ & $\gamma_{ab}$\\
    $(a-1,b+1)$ & $q_{(a,b)(a-1,b+1)}$ & $0$\\
    $(a,b)$ & $\displaystyle - \sum_{k, l \in \{-1, 0, 1\}}^{k \ne l}{q_{(a,b)(a+k,b+l)}}$ & $- (\lambda^{(1)}_{ab} + \lambda^{(2)}_{ab} + \mu^{(2)}_{ab} + \gamma_{ab})$\\
    other & $0$ & $0$
  \end{tabular}
\end{center}
for $i = (a,b)$. Competition processes are the most general bivariate Markov processes that only allow transitions between neighboring states. Many practical models in biology are special cases of these processes such as epidemic models \citep{mckendrick1926applications,Kermack1927, griffiths1972} and predator-prey models \citep{hitchcock1986, owen2014}.


\subsubsection{Sufficient condition for regularity}
\label{sec:bbd_regularity}

\begin{mydef}
A birth/birth-death process is regular if there is a unique set of transition probabilities $P^{a_0b_0}_{ab}(t)$ satisfying the system of equations (\ref{eqn:trans_eq}). 
\end{mydef}
Here, we establish the sufficient condition for regularity of a birth/birth-death process. For $k \in \mathbb N$, we denote:
 \begin{align}
 & D_k = \{ (a,b): a + b = k \} \in \mathbb{N} \times \mathbb{N},~\text{and} \nonumber \\
 & \lambda_k = \max_{(a,b) \in D_k} \{ \lambda^{(1)}_{ab} + \lambda^{(2)}_{ab} \}.
 \end{align}
 \begin{thm}
 The sufficient condition for regularity of a general birth/birth-death process is $ \sum_{k=1}^\infty{1/\lambda_k} = \infty$.
 \label{thm:reg_bbd}
 \end{thm}
 \begin{proof}
 We will apply the following Reuter's condition \citep{reuter1957}:
 \begin{lemma}
 Let ${\bf Q} = \{q_{ij}\}$ be a conservative matrix, such that $ - q_{ii} = \sum_{j \ne i}{q_{ij}} < \infty$. A continuous-time Markov chain associated with ${\bf Q}$ is regular if and only if for some $\zeta > 0$, the equation $ {\bf Q}{\bf y} = \zeta {\bf y}$
 subject to $0 \leq y_i \leq 1$ has only trivial solution ${\bf y} = {\bf 0}$.
 \label{lem:reg}
 \end{lemma}
 
 For a general birth/birth-death process, states $i$ and $j$ are in $\mathbb{N} \times \mathbb{N}$. Let $\{y_{ab}\}_{a,b \in \mathbb{N}}$  be a solution of ${\bf Q}{\bf y} = \zeta {\bf y}$ such that $y_{ab} \in [0,1]$ for any $a$ and $b$. Then, we have
 \begin{equation}
 (\zeta +  \lambda^{(1)}_{ab} + \lambda^{(2)}_{ab} + \mu^{(2)}_{ab} + \gamma_{ab}) y_{ab} = \lambda^{(1)}_{ab} y_{a+1,b} +  \lambda^{(2)}_{ab} y_{a,b+1} +  \mu^{(2)}_{ab}y_{a,b-1} + \gamma_{ab}y_{a+1,b-1} .
 \end{equation}
 Defining $y_k = \max_{(a,b) \in D_k} \{ y_{ab} \}$ and $(a_k,b_k) = \argmax_{(a,b) \in D_k} \{ y_{ab} \}$, we deduce that
 \begin{align}
  (\zeta +  \lambda^{(1)}_{a_kb_k} + \lambda^{(2)}_{a_kb_k} + \mu^{(2)}_{a_kb_k}) y_k &\leq (\lambda^{(1)}_{a_kb_k} +  \lambda^{(2)}_{a_kb_k}) y_{k+1} +  \mu^{(2)}_{a_kb_k}y_{k-1},~\mbox{and} \nonumber \\
  \zeta y_k + \mu^{(2)}_{a_kb_k}(y_k - y_{k-1}) &\leq (\lambda^{(1)}_{a_kb_k} +  \lambda^{(2)}_{a_kb_k}) (y_{k+1} - y_k).
  \end{align}
  Since $\mu^{(2)}_{a_{-1}b_{-1}} = 0$, $y_k$ is an increasing sequence. Thus,
  \begin{equation}
  \frac{\zeta}{\lambda_k} y_k \leq y_{k+1} - y_k.
 \end{equation}
 Assuming that there exists $k_0$ such that $y_{k_0} > 0$, we obtain
 \begin{equation}
 y_k \geq y_{k_0} + \zeta \sum_{i=k_0}^{k-1}{\frac{y_i}{\lambda_i}} \geq y_{k_0}(1 + \zeta \sum_{i=k_0}^{k-1}{\frac{1}{\lambda_i}}),
 \end{equation}
that is larger than $1$ if $k$ is big enough. Hence $y_k = 0$ for every $k$. Then, the theorem is proved by applying Lemma \ref{lem:reg}. 
 \end{proof}
 
 Note that the condition in Theorem \ref{thm:reg_bbd} generalizes 
 the classical regularity condition of a pure birth process \citep{feller1968introduction}. From now on, we assume that our birth/birth-death processes are regular.


\subsubsection{Recursive formula for transition probabilities}

In this section, we establish a recursion to calculate the transition probabilities $P^{a_0b_0}_{ab}(t)$ of a birth/birth-death process. Since we assume that our birth/birth-death process is regular, these transition probabilities are unique. 

We first note that $P^{a_0b_0}_{ab}(t) = 0$ for all $a<a_0$. Let $f_{ab}(s),~s \in \mathbb{C}$, be the Laplace transform of $P^{a_0b_0}_{ab}(t)$, that is
\begin{equation}
f_{ab}(s) = \mathcal{L}[P^{a_0b_0}_{ab}(t)](s) = \int_0^\infty{e^{-st}P^{a_0b_0}_{ab}(t)dt}.
\end{equation}
From (\ref{eqn:trans_eq}), we have
\begin{align}
s f_{ab}(s) - P^{a_0b_0}_{ab}(0) &= \lambda^{(1)}_{a-1,b} f_{a-1,b}(s) + \lambda^{(2)}_{a,b-1} f_{a,b-1}(s) + \mu^{(2)}_{a,b+1} f_{a,b+1}(s) \nonumber\\ 
& + \gamma_{a-1,b+1} f_{a-1,b+1}(s)  - (\lambda^{(1)}_{ab} + \lambda^{(2)}_{ab} + \mu^{(2)}_{ab} + \gamma_{ab}) f_{ab}(s),~(a,b) \in \mathbb{N}^2.
 \label{eqn:Lap}
\end{align}
Note that $f_{ab}(s)$ is the unique solution of (\ref{eqn:Lap}) by the uniqueness of $P^{a_0b_0}_{ab}(t)$. We construct the recursive approximation formulae for $f_{ab}(s)$ using continued fractions. Appendix \ref{sec:cf} provides necessary background on continued fractions and their convergents. Denote 
\begin{align}
& x_{a1} = - \frac{1}{\mu^{(2)}_{a1}};~x_{ab} = - \frac{\lambda^{(2)}_{a,b-2}}{\mu^{(2)}_{ab}},~b \geq 2 \nonumber \\
& y_{ab} = - \frac{s + \lambda^{(1)}_{a,b-1} + \lambda^{(2)}_{a,b-1} + \mu^{(2)}_{a,b-1} + \gamma_{a,b-1}}{\mu^{(2)}_{ab}},~b \geq 1,
\end{align}
and consider the following continued fraction
\begin{equation}
\phi^{(0)}_{a0}(s)  = \cfrac{x_{a1}}{
y_{a1} +
\cfrac{x_{a2}}{
y_{a2} +
\cfrac{x_{a3}}{
y_{a3} + \cdots
}~.}} 
\label{eqn:phi1}
\end{equation}
We can construct the sequence $\{ \phi^{(0)}_{ab}(s) \}_{b=0}^\infty$  (Definition \ref{def:corrseq}, Appendix \ref{sec:cf}) as follows:
\begin{align}
& (s + \lambda^{(1)}_{a0} + \lambda^{(2)}_{a0})\phi^{(0)}_{a0}(s) -  \mu^{(2)}_{a1}\phi^{(0)}_{a1}(s) = 1,~\mbox{and} \nonumber \\
& (s + \lambda^{(1)}_{a,b-1} + \lambda^{(2)}_{a,b-1} + \mu^{(2)}_{a,b-1} + \gamma_{a,b-1} )\phi^{(0)}_{a,b-1}(s) - \lambda^{(2)}_{a,b-2}\phi^{(0)}_{a,b-2}(s) - \mu^{(2)}_{ab}\phi_{ab}(s) = 0,~b \geq 2.
\label{eqn:fixa}
\end{align}
Comparing the sequences in (\ref{eqn:fixa}) with (\ref{eqn:Lap}), we deduce that $ \mathcal{L}^{-1}\left[\phi^{(0)}_{ab}(s)\right] =  P^{a_00}_{ab}(t)$. Since $P^{a_00}_{ab}(t)$ is a probability distribution, we have $\sum_{(a,b) \in \mathbb{N} \times \mathbb{N}}{P^{a_00}_{ab}(t)} = 1$. Taking the Laplace transform of the previous equation, we get $ \sum_{(a,b) \in \mathbb{N} \times \mathbb{N}}{\phi^{(0)}_{ab}(s)} = 1/s$. Hence, $\lim_{b \to \infty} \phi^{(0)}_{a_0b}(s) = 0$ for every $s > 0$. By Lemma \ref{lem:cf} (Appendix \ref{sec:cf}), $\phi^{(0)}_{a0}(s)$ converges for every $s > 0$, and
\begin{equation}
\phi^{(0)}_{ab}(s) = \prod_{i=1}^{b}{x_{ai}} \cfrac{x_{a,b+1}}{
Y_{a,b+1} + \cfrac{x_{a,b+2} Y_{ab}}{
y_{a,b+2} + \cfrac{x_{a,b+3}}{
y_{a,b+3} + \cfrac{x_{a,b+4}}{
y_{a,b+4} + \cdots
}~,}}}
\label{eqn:phi2}
\end{equation}
where $Y_{ab}$ is the denominator of the $b^{\mbox{\tiny th}}$ convergent of $\phi^{(0)}_{a0}(s)$. 

From (\ref{eqn:Lap}), we note that
\begin{equation}
(s + \lambda^{(1)}_{a_0b} + \lambda^{(2)}_{a_0b} + \mu^{(2)}_{a_0b} +\gamma_{a_0b})f_{a_0b} - \lambda^{(2)}_{a_0,b-1}f_{a_0,b-1}(s) - \mu^{(2)}_{a_0,b+1}f_{a_0,b+1}(s) = 1_{\{b = b_0\}},~b \in \mathbb{N}.
\end{equation}
By Lemma \ref{lem:sys} (Appendix \ref{sec:cf}), $f_{a_0b}(s) = \phi^{(b_0)}_{a_0b}(s)$ where
\begin{equation}
\phi^{(m)}_{ab}(s) = \begin{cases} \frac{(-1)^{m-b+1} Y_{ab}}{\mu^{(2)}_{a,m+1} \prod_{i=1}^{m+1}{x_{ai}}} \phi^{(0)}_{am}(s), & \mbox{if } b \leq m \\  \frac{- Y_{am}}{\mu^{(2)}_{a,m+1} \prod_{i=1}^{m+1}{x_{ai}}}\phi^{(0)}_{ab}(s), & \mbox{if } b \geq m. \end{cases}
\label{eqn:phi3}
\end{equation}

Next, we obtain formulae for approximating $f_{ab}(s)$ recursively assuming that we already have evaluated $f_{a-1,b}(s)$. Again, from (\ref{eqn:trans_eq}), we have
\begin{equation}
(s + \lambda^{(1)}_{ab} + \lambda^{(2)}_{ab} + \mu^{(2)}_{ab} + \gamma_{ab})f_{ab}(s) - \lambda^{(2)}_{a,b-1}f_{a,b-1}(s) - \mu^{(2)}_{a,b+1}f_{a,b+1}(s) = \lambda^{(1)}_{a-1,b}f_{a-1,b}(s) + \gamma_{a-1,b+1}f_{a-1,b+1}(s),
\label{eqn:infinitesys}
\end{equation}
for $b \in \mathbb{N}$.  We approximate $f_{ab}(s)$ by solving a truncated version of (\ref{eqn:infinitesys})
for $0 \leq b \leq B$, where $B$ is sufficiently large. The intuition of how to choose $B$ follows from the observation that we want $\sum_{a=a_0}^\infty \sum_{b=B+1}^\infty {P^{a_0b_0}_{ab}(t)}$ to be small. By Lemma \ref{lem:sys} (Appendix \ref{sec:cf}), we have the following approximation:
\begin{equation}
f_{ab}(s) \approx \sum_{m=0}^B{\left[\lambda^{(1)}_{a-1,m}f_{a-1,m}(s) + \gamma_{a-1,m+1}f_{a-1,m+1}(s)\right] \phi^{(m)}_{ab}(s)}.
\end{equation}
Therefore, the transition probabilities of a birth/birth-death process can be computed recursively using the following Theorem:

\begin{thm}
Let $\phi^{(m)}_{ab}(s)$ be defined as in (\ref{eqn:phi1}), (\ref{eqn:phi2}), and (\ref{eqn:phi3}). We have 
\begin{equation}
P^{a_0b_0}_{ab}(t) = \begin{cases} 0, & \mbox{if } a < a_0 \\  \mathcal{L}^{-1}\left[f_{ab}(s)\right](t), & \mbox{if } a \geq a_0, \end{cases}
\end{equation}
where $f_{a_0b}(s) = \phi^{(b_0)}_{a_0b}(s)$ and 
\begin{equation}
f_{ab}(s) \approx \sum_{m=0}^B{\left[\lambda^{(1)}_{a-1,m}f_{a-1,m}(s) + \gamma_{a-1,m+1}f_{a-1,m+1}(s)\right] \phi^{(m)}_{ab}(s)},~a>a_0.
\label{eqn:trunc}
\end{equation}
Here, $\mathcal{L}^{-1}(.)$ denotes the inverse Laplace transform and $B$ is the truncation level.
\label{thm:trans}
\end{thm}
If the number of {\bf type 2} particles is bounded by $B^*$, we choose $B=B^*$. 
In this case, the approximation in Theorem \ref{thm:trans} is exact.
We prove that the output of our approximation scheme (\ref{eqn:trunc}) converges to $f_{ab}(s)$ as $B$ goes to infinity in Appendix \ref{sec:conv_trunc}.
Further, the transition probability returned by Theorem \ref{thm:trans} converges to the true transition probability.
This truncation error can be bounded explicitly by extending the coupling argument in \citet{crawford2016coupling} to multivariate processes.
However, we leave it as a subject of future work because a complete treatment is beyond the scope of this paper. 


\subsubsection{Numerical approximation of the transitions probabilities}

To approximate $P^{a_0b_0}_{ab}(t)$ using Theorem \ref{thm:trans}, we need to compute two quantities: the continued fractions $\phi^{(m)}_{ab}(s)$, and the inverse Laplace transform $\mathcal{L}^{-1}\left[f_{ab}(s)\right](t)$. We efficiently evaluate the continued fractions $\phi^{(m)}_{ab}(s)$ through the modified Lentz method \citep{lentz1976, thompson1986}; see Appendix \ref{sec:lentz} for more details. This algorithm enables us to control for and limit truncation error. To approximate the inverse Laplace transform $\mathcal{L}^{-1}\left[f_{ab}(s)\right](t)$, we apply the method proposed in \citet{abate1992} using a Riemann sum: 
 \begin{equation}
\mathcal{L}^{-1}\left[f_{ab}(s)\right](t) \approx \frac{e^{H/2}}{2t} {\cal R} \left [ f_{ab} \left ( \frac{H}{2t} \right ) \right ] + \frac{e^{H/2}}{t} \sum_{k=1}^\infty{(-1)^k {\cal R} \left [ f_{ab} \left ( \frac{H + 2k \pi i}{2t} \right ) \right ]}.
 \label{eqn:invLap}
 \end{equation}
 Here ${\cal R} [z]$ is the real part of $z$ and $H$ is a positive real number. \citet{abate1992} show that the error that arises in (\ref{eqn:invLap}) is bounded by $1/(e^H - 1)$. Moreover, we can use the Levin transform \citep{levin1973} to improve the rate of convergence  because the series in (\ref{eqn:invLap}) is an alternating series when $ {\cal R} \{ f_{ab} [(H + 2k \pi i)/(2t)] \}$ have the same sign. These numerical methods have been successfully applied by \citet{crawford2012} to compute the transition probabilities of birth-death processes.  

In practice, to handle situations where $\mu^{(2)}_{ab}$ can possibly equal to $0$ for some $(a,b)$, we re-parametrize $x_{ab}$ and $y_{ab}$ as follows:
\begin{align}
& x_{a1} = 1;~x_{ab} = - \lambda^{(2)}_{a,b-2} \mu^{(2)}_{a,b-1},~b \geq 2,~\mbox{and} \nonumber \\
& y_{ab} = s + \lambda^{(1)}_{a,b-1} + \lambda^{(2)}_{a,b-1} + \mu^{(2)}_{a,b-1} + \gamma_{a,b-1},~b \geq 1.
\end{align}
With this new parametrization, we obtain
\begin{equation}
\phi^{(m)}_{ab}(s) = \begin{cases}  \cfrac{(\prod_{i=b+1}^m{\mu^{(2)}_{ai}}) Y_{ab}}{
Y_{a,m+1} + \cfrac{x_{a,m+2} Y_{am}}{
y_{a,m+2} + \cfrac{x_{a,m+3}}{
y_{a,m+3} + \cfrac{x_{a,m+4}}{
y_{a,m+4} + \cdots
}}}}~, & \mbox{if } b \leq m \\  
\cfrac{(\prod_{i=m+1}^b{\lambda^{(2)}_{ai}}) Y_{am}}{
Y_{a,b+1} + \cfrac{x_{a,b+2} Y_{ab}}{
y_{a,b+2} + \cfrac{x_{a,b+3}}{
y_{a,b+3} + \cfrac{x_{a,b+4}}{
y_{a,b+4} + \cdots
}}}}~, & \mbox{if } b \geq m. \end{cases}
\end{equation}
Our complete algorithm to compute the transition probabilities of birth/birth-death processes is implemented in the function \texttt{bbd\_prob} in a new \texttt{R} package called \texttt{MultiBD}. The function takes $t$, $a_0$, $b_0$, $\lambda^{(1)}_{ab}$, $\lambda^{(2)}_{ab}$, $\mu^{(2)}_{ab}$, $\gamma_{ab}$, $A$, $B$ as inputs and returns the transition probability matrix $\{ P^{a_0b_0}_{ab}(t) \}_{a_0 \leq a \leq A, 0 \leq b \leq B}$. Here, there is no requirement for $A$ while $B$ needs to be large enough such that $\sum_{a=a_0}^A \sum_{b=B+1}^\infty {P^{a_0b_0}_{ab}(t)}$ is small. We can check to see if $B$ is large enough by checking if $\sum_{a=a_0}^A {P^{a_0b_0}_{aB}(t)}$ is sufficiently small.

In practice, the computational complexity of evaluating each term $(f_{ab}(s))_{a_0 \leq a \leq A, 0 \leq b \leq B}$ is $\mathcal{O}((A-a_0)B^2)$ because the Lentz algorithm terminates quickly. Let $K$ be the number of iterations required by the Levin acceleration method \citep{levin1973} to achieve a certain error bound for the Riemann sum in (\ref{eqn:invLap}). Then, the total complexity of our algorithm is $\mathcal{O}((A-a_0)B^2K)$. However, evaluation of $\{ f_{ab} [(H + 2k \pi i)/(2t)] \}_{k=1}^K$ can be 
efficiently parallelized 
across different values of $k$, and we exploit this parallelism via multicore processing, delegating most of the computational work to compiled \texttt{C++} code.


\subsection{Death/birth-death processes}

Similar to the birth/birth-death process, a death/birth-death process is also a special case of competition processes. The only difference is that the number of {\bf type 1} particles is decreasing instead of increasing. Mathematically, possible transitions of a death/birth-death process ${\bf X}(t) = (X_1(t),X_2(t))$ during $(t,t+dt)$ are:
\begin{align}
& \Pr \left \{ \begin{array}{l | l} X_1(t+dt) = a-1 & X_1(t)= a \\ X_2(t+dt) = b & X_2(t) = b \end{array} \right \} = \mu^{(1)}_{ab} dt + o(dt) \nonumber \\
& \Pr \left \{ \begin{array}{l | l} X_1(t+dt) = a & X_1(t)= a \\ X_2(t+dt) = b+1 & X_2(t) = b \end{array} \right \} = \lambda^{(2)}_{ab} dt + o(dt) \nonumber \\
& \Pr \left \{ \begin{array}{l | l} X_1(t+dt) = a & X_1(t)= a \\ X_2(t+dt) = b-1 & X_2(t) = b \end{array} \right \} = \mu^{(2)}_{ab} dt + o(dt) \nonumber \\
& \Pr \left \{ \begin{array}{l | l} X_1(t+dt) = a-1 & X_1(t)= a \\ X_2(t+dt) = b+1 & X_2(t) = b \end{array} \right \} = \gamma_{ab}dt + o(dt) \nonumber \\
& \Pr \left \{ \begin{array}{l | l} X_1(t+dt) = a~~~~~ & X_1(t)= a \\ X_2(t+dt) = b & X_2(t) = b \end{array} \right \} = 1 - (\mu^{(1)}_{ab} +\lambda^{(2)}_{ab} + \mu^{(2)}_{ab} + \gamma_{ab})dt + o(dt),\end{align}
where $\mu^{(1)}_{ab} \geq 0$ is the death rate of {\bf type 1} particles given $a$ {\bf type 1} particles and $b$ {\bf type 2} particles, $\lambda^{(2)}_{ab} \geq 0$ is the birth rate of {\bf type 2} particles, $\mu^{(2)}_{ab} \geq 0$ is the death rate of {\bf type 2} particles, and $\gamma_{ab}$ is the transition rate from {\bf type 1} particles to {\bf type 2} particles. Again, we fix $\mu^{(1)}_{0,b} = \lambda^{(2)}_{a,-1} = \mu^{(2)}_{0,b}  = \gamma_{0,b} = \gamma_{a,-1} = 0$.

Following a similar argument as in Section \ref{sec:bbd_regularity}, we obtain 
a
sufficient condition for regularity of a death/birth-death process. Denote
 \begin{align}
  D_k &= \{ (a,b): a + b = k, a \leq a_0 \} \in \mathbb{N} \times \mathbb{N} \nonumber \\
  \lambda_k &= \max_{(a,b) \in D_k} \{ \lambda^{(2)}_{ab} \}  \nonumber \\
  \mu_k &= \min_{(a,b) \in D_k} \{ \mu^{(1)}_{ab} + \mu^{(2)}_{ab} \}  \nonumber \\
  \sigma_0 &=1,~ \sigma_k = \frac{\lambda_0 \ldots \lambda_{k-1}}{\mu_1 \ldots \mu_k} ,
 \end{align}
where $a_0$ is the number of {\bf type 1} particles at time $t = 0$. The following Theorem is a direct application of Theorem 1 in \citet{iglehart1964}
 \begin{thm}
 A sufficient condition for regularity of a death/birth-death process is
 \begin{equation}
 \sum_{k=0}^\infty{ \left ( \frac{1}{\lambda_k \sigma_k} \sum_{i=0}^k{\sigma_i} \right )} = \infty.
 \end{equation}
 \end{thm}
We note that if we do a transformation for a death/birth-death process ${\bf X}(t) = (X_1(t),X_2(t))$ as follows:
\begin{align}
Y_1(t) &= a_0 - X_1(t) \nonumber \\
Y_2(t) &= B - X_2(t).
\label{eqn:trans}
\end{align}
Then, ${\bf Y}(t) = (Y_1(t),Y_2(t))$ can be considered as a birth/birth-death process.
Therefore, the transition probabilities of a death/birth-death process can also be computed using the \texttt{R} function \texttt{bbd\_prob} and the transformation (\ref{eqn:trans}). Again, we want to choose $B$ such that $\sum_{a=0}^{a_0} \sum_{b=B+1}^\infty {P^{a_0b_0}_{ab}(t)}$ is small. We implement this procedure in the function \texttt{dbd\_prob} in our \texttt{R} package \texttt{MultiBD}. The function takes $t$, $a_0$, $b_0$, $\mu^{(1)}_{ab}$, $\lambda^{(2)}_{ab}$, $\mu^{(2)}_{ab}$, $\gamma_{ab}$, $A$, $B$ as inputs and returns the transition probability matrix $\{P^{a_0b_0}_{ab}(t)\}_{A \leq a \leq a_0, 0 \leq b \leq B}$. As for birth/birth-death processes, there is no requirement for $A$.


\section{Applications}
\label{sec:app}
Birth(death)/birth-death processes are appropriate for modeling two-type populations where the size of the first population is monotonically increasing (decreasing). Here we examine our methods in four applications: a within-host macro-parasite model, a birth-death-shift model for transposable elements, monomolecular reaction systems, and the stochastic SIR epidemiological model. We demonstrate that a birth (death)/birth-death process well captures the dynamics of these common biological problems, and inference using its transition probabilities often outperforms existing approximations. In particular, we emphasize that the birth (death)/birth-death process approach allows us to compute finite-time transition probabilities in the stochastic SIR model that were previously considered unknown or intractable without model simplification \citep{cauchemez2008}. 
\newcommand{\old}{\text{\tiny old}}
\newcommand{\new}{\text{\tiny new}}


\subsection{Monomolecular reaction systems} 

We illustrate the performance of our computational method by considering the following monomolecular reactions:

\begin{equation}
\begin{aligned}
&\text{Reaction} ~~R_{ab}:~~A \xrightarrow{r_{ab}} B \\
&\text{Reaction} ~~R_{ba}:~~B \xrightarrow{r_{ba}} A \\
&\text{Outflow} ~~~O_{b}:~~B \xrightarrow{o_{b}} *
\end{aligned}
\label{eqn:reactions}
\end{equation}
where $r_{ab}, r_{ba}$ is the reaction rates, and $o_{b}$ is the outflow rate.
Denote
\[
Q = \begin{pmatrix} 
- r_{ab} & r_{ba} \\
r_{ab} & -r_{ba} - o_{b} 
\end{pmatrix}, \quad
p^{(a)} = e^{Q t} 
\begin{pmatrix} 
1 \\
0
\end{pmatrix}, \quad
p^{(b)} = e^{Q t} 
\begin{pmatrix} 
0 \\
1
\end{pmatrix}.
\]
By Theorem 1 in \citet{jahnke2007solving}, the transition probabilities of the reaction system (\ref{eqn:reactions}) at time $t > 0$ is
\begin{equation}
P_{ab}^{a_0 b_0}(t) = \mathcal{M}(~.~, a_0, p^{(a)}) \star \mathcal{M}(~.~, b_0, p^{(b)})
\label{eqn:react_analytic}
\end{equation}
where $\mathcal{M}(x,N,p)$ is the multinomial distribution and $\star$ denotes the convolution operator. As analytic expressions for transition probabilities exist for this class of reactions, this example serves as a baseline for comparison to assess the accuracy of our method.

To study these processes in our framework, let $A(t)$ denote the total number of particle $A$ at time $t$ and $L(t)$ be the total number of particle $B$ leaving the system up to $t$.
Then, $\{L(t), A(t)\}$ is a birth/birth-death process with the following possible transitions during $(t, t + dt)$: 

\begin{align*}
&\Pr \left \{ \begin{array}{l | l} L(t+dt)~ = i+1  & L(t)= i \\ A(t+dt) = j~ & A(t) = j \end{array} \right \} =  o_b (a_0 + b_0 - i - j)^+ dt + o(dt), \\
&\Pr \left \{ \begin{array}{l | l} L(t+dt)~ = i  & L(t)= i \\ A(t+dt) = j+1~ & A(t) = j \end{array} \right \} =  r_{ba} (a_0 + b_0 - i - j)^+ dt + o(dt), \\
&\Pr \left \{ \begin{array}{l | l} L(t+dt)~ = i  & L(t)= i \\ A(t+dt) = j-1~ & A(t) = j \end{array} \right \} =  r_{ab}j dt + o(dt),~\mbox{and} \\
&\Pr \left \{ \begin{array}{l | l} L(t+dt)~ = i  & L(t)= i \\ A(t+dt) = j~~~~~~  & A(t) = j \end{array} \right \} =  1 - [r_{ab}j + (o_b + r_{ba}) (a_0 + b_0 - i - j)^+ ]dt + o(dt).
\end{align*}
Here $x^+ = \max(0,x)$.
Therefore, $P_{ab}^{a_0 b_0}(t)$ can be computed using our method implemented in the \texttt{R} function \texttt{bbd\_prob}.

We use \texttt{bbd\_prob} to calculate $\{P^{20,0}_{ab}(1)\}_{0 \leq a \leq 20, 0 \leq b \leq 20}$ of the reaction system (\ref{eqn:reactions}) with $r_{ab} = 2, r_{ba} = 0.5$ and $o_b = 1$. 
The $L_1$ distance between our result and the analytic result (\ref{eqn:react_analytic}) is less than $4.7 \times 10^{-9}$, thus confirming the accuracy of our method compared to explicit analytic solutions.


\subsection{Birth-death-shift model for transposable elements} 
Transposable elements or transposons are genomic sequences that can either duplicate, with a new copy moving to a new genomic location, move to a different genomic location, or be deleted from the genome.  \citet{rosenberg2003} model the number of copies of a particular transposon using a linear birth-death-shift process; a birth is a duplication event, a death is a deletion event, and shift is a switching position event. \citet{Xu2015} propose representing this birth-death-shift process by a linear multi-type branching process ${\bf X}(t) = (X_{\old}(t), X_{\new}(t))$ tracking the number of occupied sites where $X_{\old}(t)$ is the number of initially occupied sites and $X_{\new}(t)$ is the number of newly occupied sites. Let $\lambda$, $\mu$, and $\nu$ be the birth, death, and shift rates respectively. The transitions of ${\bf X}(t)$ during a small time interval occur with probabilities
\begin{align}
& \Pr \left \{ \begin{array}{l | l} X_{\old}(t+dt) = x_{\old} - 1~ & X_{\old}(t)~ = x_{\old} \\ X_{\new}(t+dt) = x_{\new} & X_{\new}(t) = x_{\new} \end{array} \right \} = (\mu x_{\old}) dt + o(dt), \nonumber \\
& \Pr \left \{ \begin{array}{l | l} X_{\old}(t+dt) = x_{\old} & X_{\old}(t)~ = x_{\old} \\ X_{\new}(t+dt) = x_{\new} - 1 & X_{\new}(t) = x_{\new} \end{array} \right \} = (\mu x_{\new})dt + o(dt), \nonumber \\
& \Pr \left \{ \begin{array}{l | l} X_{\old}(t+dt) = x_{\old} & X_{\old}(t)~ = x_{\old} \\ X_{\new}(t+dt) = x_{\new} + 1 & X_{\new}(t) = x_{\new} \end{array} \right \} = \lambda (x_{\old} + x_{\new}) dt + o(dt), \nonumber \\
& \Pr \left \{ \begin{array}{l | l} X_{\old}(t+dt) = x_{\old} - 1 & X_{\old}(t)~ = x_{\old} \\ X_{\new}(t+dt) = x_{\new} + 1 & X_{\new}(t) = x_{\new} \end{array} \right \} = (\nu x_{\old}) dt + o(dt),~\mbox{and} \nonumber \\
& \Pr \left \{ \begin{array}{l | l} X_{\old}(t+dt) = x_{\old}~~~~~~~ & X_{\old}(t)~ = x_{\old} \\ X_{\new}(t+dt) = x_{\new} & X_{\new}(t) = x_{\new} \end{array} \right \} =  1 - (\mu + \lambda + \nu) x_{\old} - (\mu + \lambda) x_{\new} dt + o(dt).
\end{align}
Equivalent to the branching process representation, notice that in this case ${\bf X}(t)$ is also a death/birth-death process. Hence, we can effectively compute its transition probabilities. In contrast, \citet{Xu2015} consider the probability generating function
\begin{equation}
\Phi_{a_0b_0}(t, s_1, s_2) = {\mathbb E} \left (s_1^{X_{\old}(t)} s_2^{X_{\new}(t)} | X_{\old}(0) = a_0, X_{\new}(0) = b_0 \right ) = \sum_{a = 0}^\infty \sum_{b=0}^\infty {P^{a_0b_0}_{ab}(t) s_1^a s_2^b},
\end{equation}
where
\begin{equation}
P^{a_0b_0}_{ab}(t) = \Pr \left \{ \begin{array}{l | l} X_{\old}(t)~ = a ~~ & X_{\old}(0)~ = a_0 \\ X_{\new}(t) = b & X_{\new}(0) = b_0 \end{array} \right \}.
\end{equation}
Because of the model-specific linearity in terms of $a$ and $b$ of the birth and death rates, one can evaluate $\Phi_{jk}(t, s_1, s_2)$ by solving an ordinary differential equation. Further transforming $s_1 = e^{2 \pi i w_1}$, $s_2 = e^{2 \pi i w_2}$, the generating function becomes a Fourier series
\begin{equation}
\Phi_{a_0b_0}(t, e^{2 \pi i w_1}, e^{2 \pi i w_2}) =  \sum_{a = 0}^\infty \sum_{b=0}^\infty {P^{a_0b_0}_{ab}(t)e^{2 \pi i a w_1}e^{2 \pi i b w_2}}.
\end{equation}
Therefore, \citet{Xu2015} retrieve the transition probabilities through approximating the integral as a Riemann sum
\begin{align}
P^{a_0b_0}_{ab}(t) &= \int_0^1 \int_0^1 {\Phi_{a_0b_0}(t, e^{2 \pi i w_1}, e^{2 \pi i w_2})e^{-2 \pi i a w_1} e^{-2 \pi i b w_2} dw_1 dw_2} \nonumber \\
&\approx \frac{1}{H^2} \sum_{u=0}^{H-1} \sum_{v=0}^{H-1} {\Phi_{jk}(t, e^{2 \pi i u/H}, e^{2 \pi i v/H})e^{-2 \pi i a u/H} e^{-2 \pi i b v/H}} ,
\end{align}
and show that choosing $H$ as the smallest power of $2$ greater than $\max(a,b)$ produces accurate estimates of the true transition probabilities of the model. The authors implement this method in the \texttt{R} package \texttt{bdsem}. Using their method, evaluating $\{P^{a_0b_0}_{ab}(t)\}_{0 \leq a,b \leq H}$ requires numerically solving $H^2$ linear ordinary differential equations (ODEs). 
We perform a simulation to compare the performance between \texttt{bdsem} and our function \texttt{dbd\_prob}.
Because \citet{Xu2015} already provide a thorough empirical validation that \texttt{bdsem} produces accurate transition probabilities compared to Monte Carlo estimates from the true model, we consider a comparison to their method and omit a complete reproduction of their simulation study.  Using both routines to compute the transition probabilities of a birth-death-shift process with rates $\lambda = 0.0188$, $\mu = 0.0147$, $\nu = 0.00268$ (estimated from the IS\textit{6110} data by \citet{rosenberg2003}) repeatedly over one hundred trials leads to a negligible difference in estimated probabilities. Specifically, we computed $\{P^{10,0}_{ab}(t)\}_{0 \leq a \leq 10, 0 \leq b \leq 50}$ at three different observation period lengths $t = 1, 5, 10$, and found that the $L_1$ distance between probabilities estimated by each method is less than $4 \times 10^{-8}$ across all cases. Here, the $L_1$ distance between two matrices ${\bf U} = (u_{ij})$ and ${\bf V} = (v_{ij})$ are defined as $\| {\bf U} - {\bf V} \| = \sum_{i,j}{|u_{ij} - v_{ij}|}$.

Having validated the accuracy of our approach, we turn to a runtime comparison. The ratios of CPU time required using \texttt{bdsem} compared to \texttt{dbd\_prob} are summarized in Figure \ref{fig:vs_generating}, and note that this result is obtained using a single-thread option for \texttt{dbd\_prob}. 
We see that \texttt{dbd\_prob}  is about $15$ to $30$ times faster than the \texttt{bdsem} implementation, while producing very similar results. 

While there is a large performance difference in wall clock time, we cannot immediately conclude that our method is faster then the method in \citet{Xu2015} because computation time may depend heavily on implementation. Nonetheless, we can make some remarks about the performance of both methods that are platform-independent. Notably, the \texttt{bdsem} implementation grows slower as $t$ increases while \texttt{dbd\_prob} does not. This is expected because solving ODEs is slower when the domain increases. However, it is worth mentioning that we can use the solution paths to get the solutions of these ODEs at other time points in the domain. For example, when we solve the ODEs at $t = 10$, we also get the solutions at $t = 1$ and $5$ for free. This point becomes important in applications where we need to compute the transition probabilities at several time points. Another downside of \texttt{bdsem} is that it computes $\{P^{10,0}_{ab}(t)\}_{0 \leq a,b \leq 50}$ instead of evaluating $\{P^{10,0}_{ab}(t)\}_{0 \leq a \leq 10, 0 \leq b \leq 50}$ directly as is done by \texttt{dbd\_prob}.

\begin{figure}[!h]
\begin{center}
\includegraphics[width=.4\textwidth]{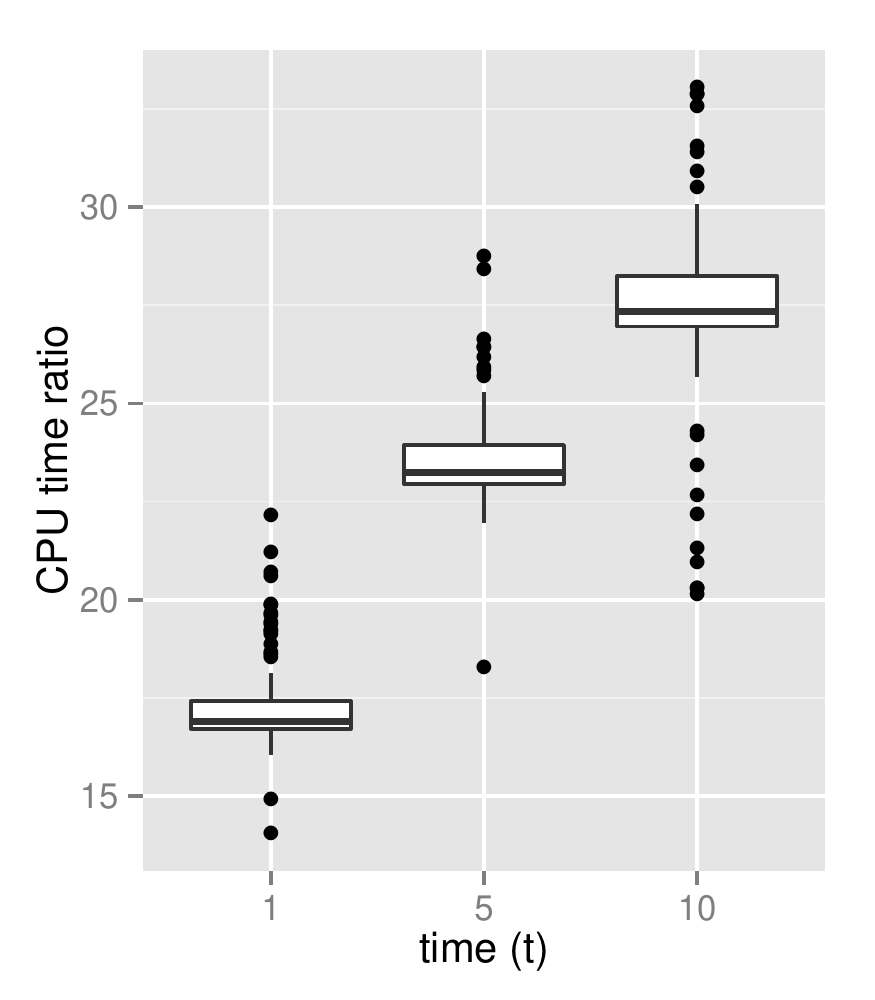}
\end{center}
\caption{CPU compute time ratios  of \texttt{bdsem} to \texttt{dbd\_prob} over 100 replications.}
\label{fig:vs_generating}
\end{figure}


\subsection{Within-host macro-parasite model}
\citet{riley2003} posit a stochastic model to describe a within-host macro-parasite population where \textit{Brugia pahangi} is the parasite and \textit{Felis catus} is the host. \textit{Brugia pahangi} is closely related to \textit{Brugia malayi} which infects millions of people in South and Southeast Asia. The model tracks the number of \textit{B. pahangi} larvae $L(t)$, the number of mature parasites $M(t)$, and hosts experience of infection $I(t)$ at time $t$. The dynamics of  $\{ L(t), M(t), I(t) \}$ follow a system of differential equations:
\begin{align}
\frac{dL}{dt}(t) & = - \mu_L L(t) - \beta I(t) L(t) - \gamma L(t), \nonumber \\
\frac{dM}{dt}(t) & = \gamma L(t) - \mu_M M(t),~\mbox{and} \nonumber\\
\frac{dI}{dt}(t) & = \nu L(t) - \mu_I I(t)
\end{align} 
where $\mu_L$ is the natural death rate and $\gamma$ is the maturation rate of larvae; $\beta$ is the death rate of larvae due to the immune response from the host; $\mu_M$ is the death rate of mature parasites; $\nu$ is the acquisition rate and $\mu_I$ is the loss rate of immunity.

\citet{drovandi2011} propose a simplification of this model by applying a pseudoequilibrium assumption for immunity, such that the immunity is constant over time. Under this pseudoequilibrium assumption, the dynamics of $\{L(t), M(t)\}$  becomes
\begin{align}
\frac{dL}{dt}(t) & = - \mu_L L(t) - \eta [L(t)]^2 - \gamma L(t),~\mbox{and} \nonumber \\
\frac{dM}{dt}(t) & = \gamma L(t) - \mu_M M(t)
\label{eqn:deterministic_parasite}
\end{align} 
where $\eta = \beta \nu/\mu_I$. We illustrate the dynamic of \eqref{eqn:deterministic_parasite} in Figure \ref{fig:dynamic}. The corresponding stochastic formulation of this model is:
\begin{align}
&\Pr \left \{ \begin{array}{l | l} L(t+dt)~ = i-1  & L(t)= i \\ M(t+dt) = j+1~ & M(t) = j \end{array} \right \} =  (\gamma i) dt + o(dt), \nonumber \\
&\Pr \left \{ \begin{array}{l | l} L(t+dt)~ = i-1  & L(t)= i \\ M(t+dt) = j~~~~~~ & M(t) = j \end{array} \right \} =  (\mu_L i + \eta i^2) dt + o(dt), \nonumber \\
&\Pr \left \{ \begin{array}{l | l} L(t+dt)~ = i  & L(t)= i \\ M(t+dt) = j-1~ & M(t) = j \end{array} \right \} =  (\mu_M j) dt + o(dt),~\mbox{and} \nonumber \\
&\Pr \left \{ \begin{array}{l | l} L(t+dt)~ = i  & L(t)= i \\ M(t+dt) = j~~~~~~  & M(t) = j \end{array} \right \} =  1 - (\gamma i + \mu_L i + \eta i^2  + \mu_M j)dt + o(dt).
\end{align}
Notably, $\{L(t), M(t)\}$ follow a death/birth-death process. 

\begin{figure}[!h]
\begin{center}
\includegraphics[width=0.7\textwidth]{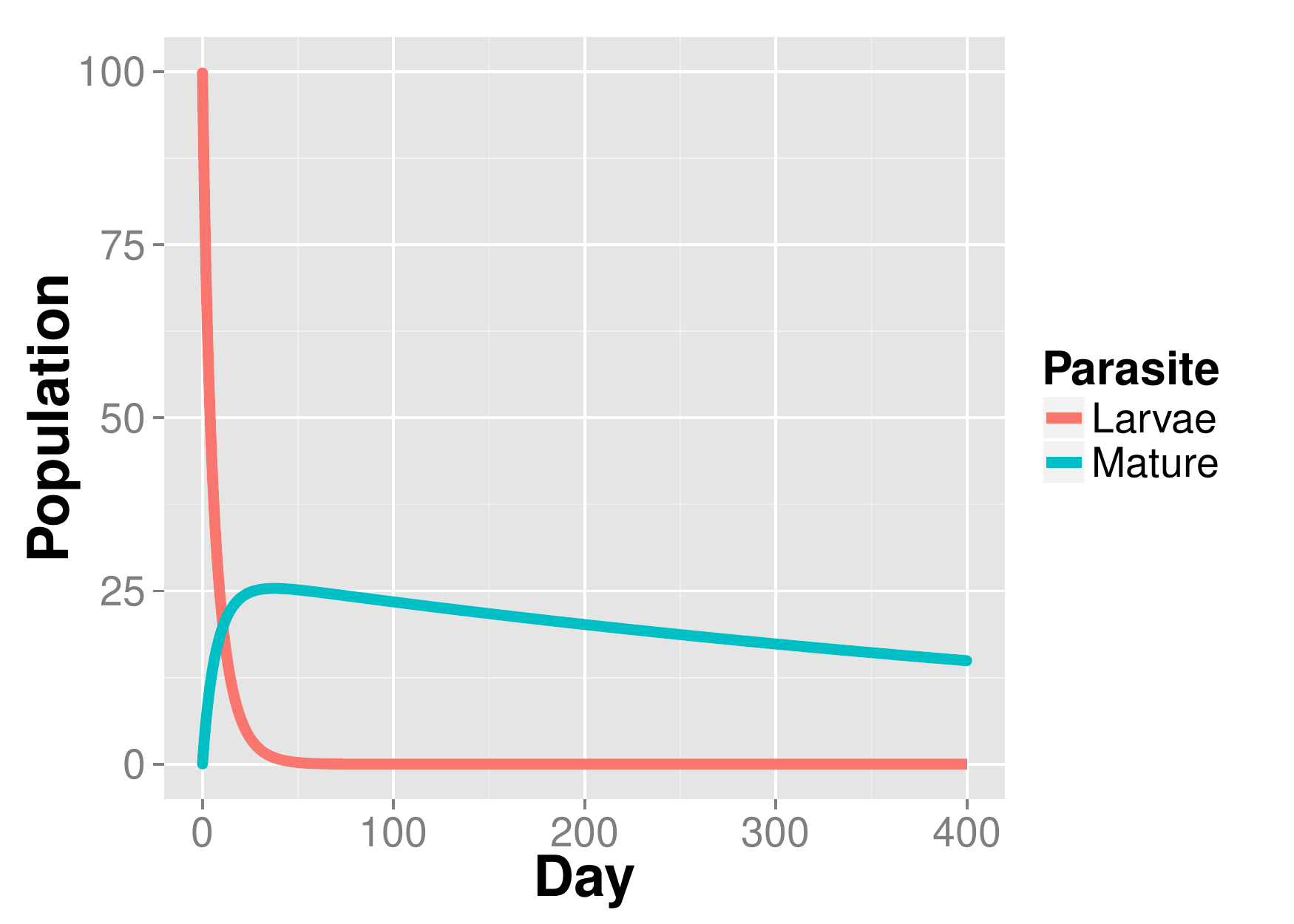}
\end{center}
\caption{The dynamic of $\{L(t), M(t)\}$ under the deterministic model \eqref{eqn:deterministic_parasite} with $\mu_L = 0.0682, \mu_M = 0.0015, \eta = 0.0009, \gamma = 0.04$ and $\{L(0), M(0)\} = \{100, 0\}$.}
\label{fig:dynamic}
\end{figure}

For this model, $\gamma$ and $\mu_M$ has been estimated at $0.04$ and $0.0015$ previously \citep[see][for more details]{drovandi2011}. To estimate the remaining parameters, \citet{drovandi2011} examine the number of mature parasites at host autopsy time (at most 400 days) of those injected with approximately $100$ juveniles,
assume \textit{a priori} $\mu_L$ and $\eta$ are uniform[0,1) and apply ABC to draw inference because the traditional matrix exponentiation method is computationally prohibitive here. The basic idea of ABC involves sampling from an approximate posterior distribution
\begin{equation}
f(\boldsymbol{\theta}, Y | \rho(Y,Y_s) \leq \epsilon) \propto f(Y_s | \boldsymbol{\theta}) \pi(\boldsymbol{\theta}) 1_{ \rho(Y,Y_s) \leq \epsilon},
\end{equation}
where $\boldsymbol{\theta}$ is the vector of unknown parameters, $\epsilon > 0$ is an \emph{ad hoc} tolerance, and $\rho(Y,Y_s)$ is a discrepancy measure between summary statistics of the observed data $Y$ and the simulated data $Y_s$. Because the sufficient statistics are not available for this problem, the authors use a goodness-of-fit statistic. However, the ABC method suffers from 
loss of information
because of non-zero tolerance and non-sufficient summary statistics \citep{sunnaaker2013approximate}. Therefore, credible intervals obtained by the ABC approach are potentially inflated \citep{csillery2010approximate}. 

In contrast, our method makes direct likelihood computation and in turn evaluation of the posterior density feasible. Figure \ref{fig:logliksurf} displays a visualization of the posterior density surface 
of  $(\log \mu_L, \log \eta)$ computed using our method,
given the collection of numbers of mature parasites $M(t)$ at autopsy under this model \citep[see][for more details about the data]{drovandi2011}.
Importantly in this example, we are able to efficiently integrate out the unobserved larvae counts $L(t)$ at autopsy. The approximate estimate obtained by \citet{drovandi2011} using ABC is overlaid on this density surface for comparison, and does not align with the highest density region of our computed posterior.
Note that the posterior is flat when $\eta$ is close to $0$, and has an unusual tail toward the region where the ABC estimate lies. This suggests that the previous ABC approach fails to explore the region with high posterior probability well, likely due to loss of information incurred by the method, resulting in a poor estimate from the data.

\begin{figure}[!h]
\begin{center}
\includegraphics[width=0.7\textwidth]{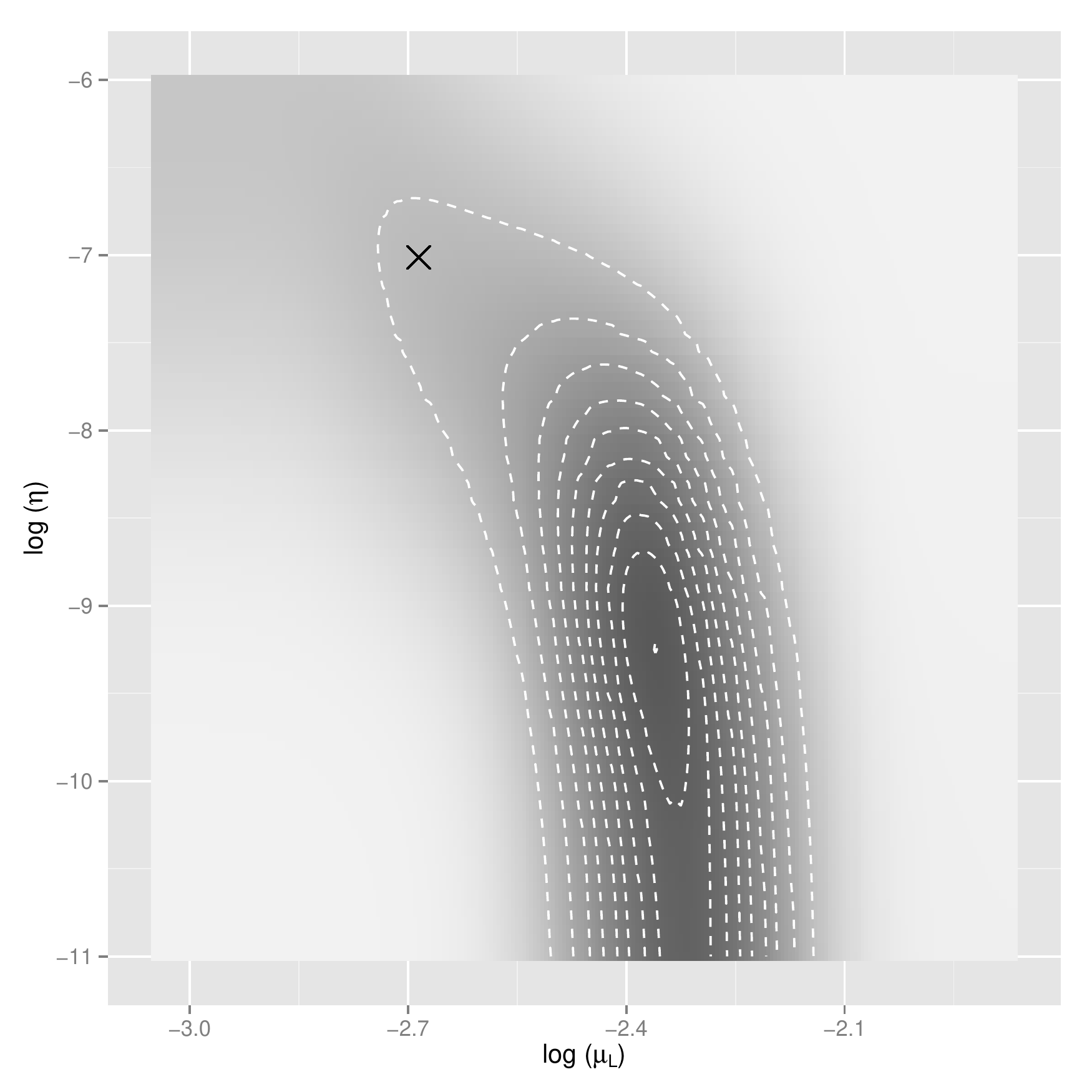}
\end{center}
\caption{Posterior density surface of $(\log \mu_L, \log \eta)$ under within-host macro-parasite model. The ``$\times$" symbol represents the estimate from  \citet{drovandi2011} using the ABC method.}
\label{fig:logliksurf}
\end{figure}

Finally, we consider this example toward a second runtime comparison between our method and 
\texttt{Expokit}, a state-of-the-art matrix exponentiation package with efficient implementation.
In particular, we compute the transition probability matrix $\{P^{100,0}_{ij}(t)\}_{0 \leq i \leq 100, 0 \leq j \leq 100}$ of $\{L(t), M(t)\}$ with $\mu_L = 0.0682, \mu_M = 0.0015, \eta = 0.0009, \gamma = 0.04$ at $t = 100, 200, 400$ using our function \texttt{dbd\_prob} and the function \texttt{expv} in \texttt{expoRkit}, an \texttt{R}-interface to the Fortran package \texttt{Expokit}.\
Both methods produce similar results: the $L_1$ distance between the two estimated transition probability matrices is less than $3 \times 10^{-9}$ across all cases.
In terms of speed, we see that \texttt{dbd\_prob} is roughly twice as fast as \texttt{expv} when $t=100, 200$, but about $9$-fold faster when $t=400$ (Figure \ref{fig:vs_expM}).
It is worth mentioning that \texttt{dbd\_prob} can be further accelerated via parallelization.

\begin{figure}[!h]
\begin{center}
\includegraphics[width=.4\textwidth]{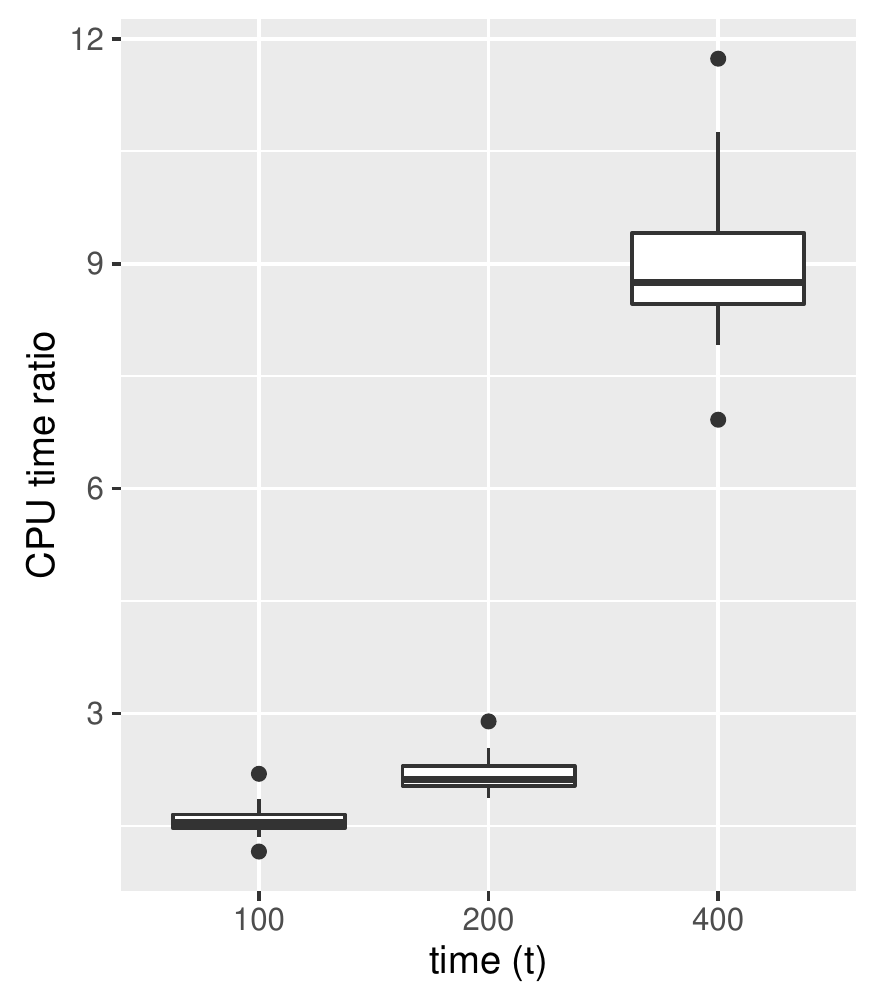}
\end{center}
\caption{CPU compute time ratios  of \texttt{expv} to \texttt{dbd\_prob} over 100 replications.}
\label{fig:vs_expM}
\end{figure}


\subsection{Stochastic SIR model in epidemiology}

\citet{mckendrick1926applications} models the spread of an infectious disease in a closed population by dividing the population into three categories: susceptible persons ($S$), infectious persons ($I$) and removed persons ($R$). Since the population is closed, the 
total population size $N$ obeys the conservation equation $N = S(t) + I(t) + R(t)$ for all time $t$. The deterministic dynamics of these three subpopulations follow a system of nonlinear ordinary differential equations \citep{Kermack1927}:
\begin{align}
\frac{dS}{dt}(t) &= - \beta S(t)I(t), \nonumber \\
\frac{dI}{dt}(t) &= \beta S(t)I(t) - \alpha I(t),~\mbox{and} \nonumber \\
\frac{dR}{dt}(t) &= \alpha I(t),
\end{align}
where $\alpha > 0$ is the removal rate and $\beta > 0$ is the infection rate of the disease. This system of equations cannot be solved analytically, but we can obtain its solution numerically. 
An important quantity for the SIR model is the basic reproduction number $R_0 = \beta N / \alpha$ \citep{earn2008light}. This quantity determines whether a spread of an infectious disease becomes an epidemic. In particular, an epidemic can only occur when $R_0 > 1$.

Unfortunately, the deterministic model is not suitable when the community is small \citep{britton2010}. In these situations, the original stochastic SIR model \citep{mckendrick1926applications} becomes more appropriate. Moreover, \citet{andersson2000stochastic} argue that stochastic epidemic models are preferable when their analysis is possible because (1) stochastics are the most natural way to describe a spread of diseases, (2) some phenomena do not satisfy the law of large numbers and can only be analyzed in the stochastic setting (for example, the extinction of endemic diseases only occurs when the epidemic process deviates from its expected value), and (3) quantifying the uncertainty in estimates requires stochastic models. Nonetheless, one can bypass Andersson and Britton's third argument by imposing random sampling errors around the deterministic compartments. Therefore, it is important to distinguish between the deterministic SIR model with sampling errors and the stochastic SIR model.

Without loss of generality, the stochastic SIR model needs only track $S(t)$ and $I(t)$ because $S(t) + I(t) + R(t)$ remains constant. All possible transitions of $\{S(t),I(t)\}$ during a small time interval $(t,t + dt)$ occur with probabilities
\begin{align}
& \Pr \left \{ \begin{array}{c | c} S(t+dt) = s~~~~~  & S(t)= s \\ I(t+dt) = i-1 & I(t) = i \end{array} \right \} = (\alpha i) dt + o(dt), \nonumber \\
& \Pr \left \{ \begin{array}{c | c} S(t+dt)= s - 1 & S(t)= s \\ I(t+dt) = i+1 & I(t) = i \end{array} \right \} = (\beta s i) dt + o(dt),~\mbox{and} \nonumber \\
& \Pr \left \{ \begin{array}{c | c} S(t+dt)= s~~~~~ & S(t)= s \\ I(t+dt) = i~~~~~ & I(t) = i \end{array} \right \} = 1 - (\alpha i + \beta s i)dt + o(dt).
\label{eqn:sir}
\end{align}
We see that $\{S(t),I(t)\}$ is a death/birth-death process with $\mu^{(1)}_{si} = \lambda^{(2)}_{s,i} = 0$, $\mu^{(2)}_{s,} = \alpha i $, $\gamma_{si} = \beta s i$. 

Due to the interaction between populations and nonlinear nature of the model, mechanistic analysis of the stochastic SIR model is difficult, and the lack of an expression for transition probabilities has been a bottleneck for statistical inference. \cite{renshaw2011} remarks 
that while one can write out the Kolmogorov forward equation for the system, the ``associated mathematical manipulations required to generate solutions can only be described as heroic." 
Instead, the majority of efforts involve either simulation based methods or simplifications and tractable approximations to the SIR model. For instance, the stochastic SIR model can be analyzed using ABC \citep{Mckinley2009},
but we have already mentioned limitations of this approach. 
Particle filter methods can be used to analyze SIR models within maximum likelihood \citep{Ionides2006, Ionides2015} and Bayesian frameworks \citep{Andrieu2010, Dukic2012}, 
but these methods are computationally very demanding and often suffer from convergence problems.
When examining large epidemics, to make the likelihood tractable it is reasonable to apply a continuous approximation to the large populations, modeled as a diffusion process with exact solutions \citep{cauchemez2008}. However, such an approach is a poor proxy for the SIR model when observed counts are low. When data are collected at regular intervals and coincide with disease generation timescales, it is also possible to study discrete-time epidemic models--- the time-series SIR (TSIR) model is one well-known example \citep{finkenstadt2000}. However, these simplifications also have their shortcomings, relying on the relatively strong assumption that populations are constant over each interval between observation times. 


In the death/birth-death framework, our method enables practical computation of these quantities \textit{without} any simplifying model assumptions. In Section \ref{sec:eyam}, we will apply our method to analyze the population of Eyam during the plague of 1666 \citep{raggett1982} to estimate the infection and the death rates of this disease, using the death/birth-death transition probabilities within a Metropolis-Hastings algorithm. Here, we first examine the accuracy of these transition probabilities themselves. We compare the continued fraction method to empirical transition probabilities obtained via simulation from the true model as ground-truth, and to a new two-type branching approximation to the SIR model introduced below. The branching process approximation is appropriate when transition probabilities need to be computed for short time intervals, and its simple expressions for transition probabilities enable much more efficient computation. However, we show that as transition time intervals increase, the branching approximation becomes less accurate, while the transition probabilities computed under the death/birth-death model remain very accurate.

While branching processes fundamentally rely on \textit{independence} of each member of the population, we can nonetheless make a fair approximation by mimicking the interaction effect of infection over short time intervals. 
In the branching model, let $X_1(t)$ denote the susceptible population and $X_2(t)$ denote the infected population at time $t$, with details and derivation included in Appendix \ref{ap:sir}. Over any time interval $[t_0, t_1)$, we use the initial population $X_2(0)$ as a constant scalar for the instantaneous rates.
This branching process model has instantaneous infection rate $\beta X_2(0) X_1(t)$ and recovery rate $\alpha X_2(t)$ for all $t \in [t_0, t_1)$, closely resembling the true SIR model rates, with the exception of fixing $X_2(0)$ in place of $X_2(t)$ in the rate of infection. This constant initial population fixes a piecewise homogeneous per-particle birth rate to satisfy particle independence while mimicking interactions, but notice that \textit{both} populations can change over the interval, offering much more flexibility than models such as TSIR that assume constant populations and rates between discrete observations.

This branching model admits closed-form solutions to the transition probabilities that can be evaluated quickly and accurately. The transition probabilities of the two-type branching approximation to the SIR model over any time interval of length $t$ are given by
\begin{align}
\label{eq:SIRtrans}  
\text{Pr} \left\{ \mb{X}(t+\tau) = (k,l) | \mb{X}(\tau) = (m,n) \right\} := P_{kl}^{mn}(t) = \sum_{i=0}^l {l \choose i} A(l-i) B(i),
\end{align} where
\begin{align}
\label{eq:conditionB}
B(i) &= 0 \text{ for all } i \geq n, \text{ otherwise}, \nonumber \\
B(i) &= \frac{n!}{(n-i)!} (1 - e^{-\alpha t})^{n-i} e^{-i \alpha t} 
\end{align}
and
\begin{align}
\label{eq:conditionA}
A(l-i) &= 0 \text{ for all } (l-i) \geq (m-k), \text{  otherwise}, \nonumber \\
A(l-i) &= \frac{m!}{(m-k-(l-i))!} e^{-k \beta n t} \left[ 1 - \frac{\beta n }{\beta n - \alpha} e^{-\alpha t} - \left(1 - \frac{ \beta n}{\beta n - \alpha} \right) e^{-\beta n t} \right] ^{m-k-(l-i)} \nonumber \\
& \quad \quad \times \left[ \frac{ \beta n}{\beta n - \alpha} ( e^{- \alpha t} - e^{-\beta n t} ) \right]^{l-i}.  
\end{align}
The sum over products of expressions \eqref{eq:conditionB} and \eqref{eq:conditionA} in equation \eqref{eq:SIRtrans} may look unwieldy, but this sum is computed extremely quickly with a vectorized implementation, and with high degrees of numerical stability. 
In settings when such a model is appropriate and $(X_1(t), X_2(t)) \approx (S(t), I(t) )$, the branching approximation can offer a much more computationally efficient alternative to the continued fraction method. 

\subsection{Transition probabilities of the SIR model}
\begin{figure}
\centering 
    \includegraphics[scale=1.0]{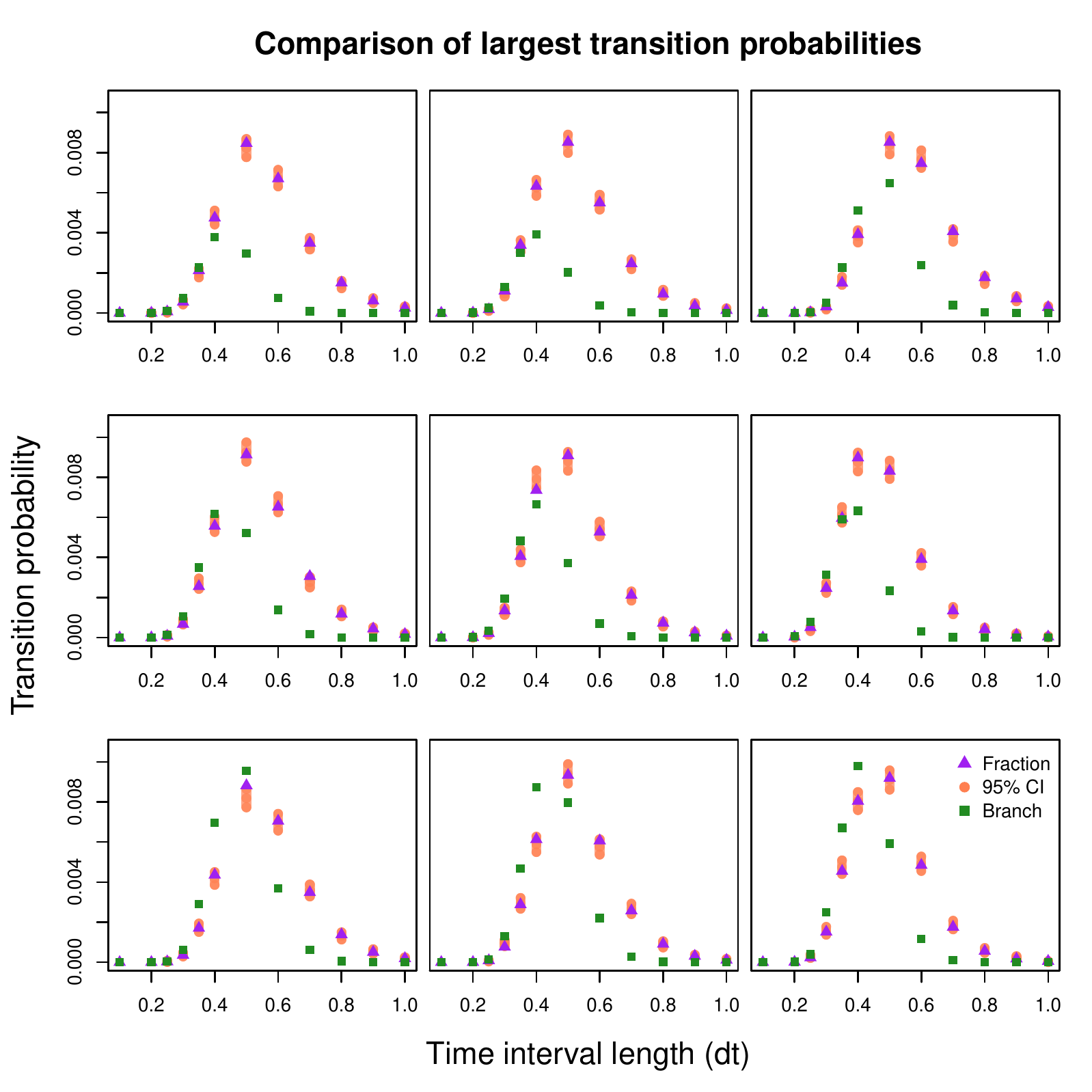}
   \caption{ The plot above displays the values of the nine largest transition probabilities when $t=0.5$ as we vary $t$ from $0.1, \ldots, 1.0$. Parameters used to generate data are initialized at $I_0 = 15, S_0 = 110, \alpha = 3.2, \beta = 0.025$.  Empirical Monte Carlo $95\%$ confidence intervals over $150,000$ simulations from the true model are depicted in orange. Probabilities computed using the continued fraction expansion are depicted by purple triangles, while probabilities computed under the branching approximation are denoted by green squares.} 
   \label{fig:transCompare}
\end{figure}

Figure \ref{fig:transCompare} provides a comparison between methods of computing transition probabilities. Included are transition probabilities corresponding to the nine pairs of system states $\left\{ (m,n),(k,l) \right\}_j$, $j = 1, \ldots, 9$, such that $P_{kl}^{mn}(0.5)$ is largest. Fixing these indices, we plot the set of probabilities $\left\{P_{kl}^{mn}(t)\right\}$ while varying $t$ between $0.1$ and $1.0$. We see that transition probabilities computed using the continued fraction method under the death/birth-death model very closely match those computed empirically via simulation from the model, taken to be the ground truth.  Almost all such probabilities in Figure \ref{fig:transCompare} fall within the 95\% confidence interval, while the branching process transitions follow a similar shape over time, but fall outside of the confidence intervals for many observation intervals. An additional heatmap visualization comparing the support of transition probabilities is included in the Appendix, and shows that the branching approximation is accurate with similar support to the empirical transition probabilities for a shorter time interval of length $t = 0.5$, but becomes visibly further from the truth when we increase the observation length to $t=1.0$.

\section{The Plague in Eyam revisited}
\label{sec:eyam}
We revisit the outbreak of plague in Eyam, a village in the Derbyshire Dales district, England,  over the period from June 18th to October 20th, 1666. This plague outbreak is widely accepted to originate from the Great Plague of London, that killed about $15\%$ of London's population at that time. To prevent further spread of the plague after infestation, the Eyam villagers did not escape the village, instead isolating themselves from the outside world. At the end of this horrific event, only 83 people had survived out of an initial population of 350. We summarize data recording the spread of the disease \citep{raggett1982} in Table \ref{tab:eyam}. As mentioned in \citet{raggett1982}, this data are obtained by counting the number of deaths from the dead list and estimating the infective population from the list of future deaths assuming a fixed length of illness prior to death. Then, the susceptible population can be computed easily because the the town is isolated.
\begin{table}[h]
\begin{center}
\begin{tabular}{ lcccccccc }
\cline{2-9}
& \multicolumn{8}{c}{Time (months)} \\
\cline{2-9}
& 0 & 0.5 & 1 & 1.5 & 2 & 2.5 & 3 & 4 \\
\hline
Susceptible population & 254 & 235 & 201 & 153 & 121 & 110 & 97 & 83 \\
Infective population & 7 & 14 & 22 & 29 & 20 & 8 & 8 & 0 \\
\hline
\end{tabular}
\end{center}
\caption{Susceptible and infectious population size in Eyam from June 18th to October 20th, 1666.}
\label{tab:eyam}
\end{table}

\noindent
\citet{raggett1982} analyzes these data using the stochastic SIR model (\ref{eqn:sir}). In this model, $\alpha$ is the unknown death rate of infective people and $\beta$ is the unknown infection rate of the plague. The author uses a simple approximation method for the forward differential equation and comes up with a point estimate $(\hat \alpha, \hat \beta) = (3.39, 0.0212)$. We take a Bayesian approach to re-analyze these data.

With $n$ observations $\{ (s_k, i_k) \}_{k=1}^n$ at time $\{ t_k \}_{k=1}^n$, the log of the likelihood function is:
\begin{equation}
\log l \left( \alpha, \beta \big | \{ (s_k, i_k) \}_{k=1}^n \right ) = \sum_{k=1}^{n-1}{ \log \Pr \left \{ \begin{array}{l | l} S(t_{k+1}) = s_{k+1}  & S(t_k) = s_k \\ I(t_{k+1}) = i_{k+1} & I(t_k) = i_k \end{array} \right \} }.
\label{eqn:logllh}
\end{equation}
Because $\{S(t),I(t)\}$ is a death/birth-death process, 
the individual transition probabilities
can be computed efficiently using our continued fraction method. Hence, the log of the likelihood (\ref{eqn:logllh}) can be computed easily. Since $\alpha$ and $\beta$ are non-negative, we opt to use $\log \alpha$ and $\log \beta$ as our model parameters and assume \textit{a priori} that $\log \alpha \sim {\cal N}(\mu = 0, \sigma = 100)$ and $\log \beta \sim {\cal N}(\mu = 0, \sigma = 100)$. We explore the posterior distribution of $(\log \alpha, \log \beta)$ using a random-walk Metropolis algorithm implemented in the \texttt{R} function \texttt{MCMCmetrop1R} from package \texttt{MCMCpack} \citep{martin2011}. We start the chain from Raggett's estimated value $(\log(3.39), \log(0.0212))$ and run it for $100000$ iterations. We discard the first $20000$ iterations and summarize the posterior distribution of $(\alpha, \beta)$ using the remaining iterations. We illustrate the density of this posterior distribution in Figure \ref{fig:density}. The posterior mean of $\alpha$ is $3.22$ and the $95\%$ Bayesian credible interval for $\alpha$ lies in $(2.69, 3.82)$. Those corresponding quantities for $\beta$ are $0.0197$ and $(0.0164, 0.0234)$. Notice that our credible intervals include the point estimate $(\hat \alpha, \hat \beta) = (2.73, 0.0178)$ from \citet{brauer2008compartmental} using the deterministic SIR model and Raggett's point estimate $(\hat \alpha, \hat \beta) = (3.39, 0.0212)$. 

We also apply the two-type branching approximation to compute the log of the likelihood (\ref{eqn:logllh}). Using the same random-walk Metropolis algorithm as before, we explore the posterior distribution of $(\alpha, \beta)$ and visualize it in Figure \ref{fig:TransProb_density}. The posterior mean of $\alpha$ is $3.237$ and the $95\%$ Bayesian credible interval for $\alpha$ is $(2.7, 3.84)$, while those quantities for $\beta$ are $0.02$ and $(0.0171, 0.023)$. Although the posterior means and the $95\%$ Bayesian credible intervals are similar to ones from the continued fraction method, we see in Figure \ref{fig:TransProb_density} that this method fails to fully capture the posterior correlation structure between $\alpha$ and $\beta$.

\begin{figure}[!h]
\subfigure[Continued fraction method]{
\includegraphics[width=0.5\textwidth]{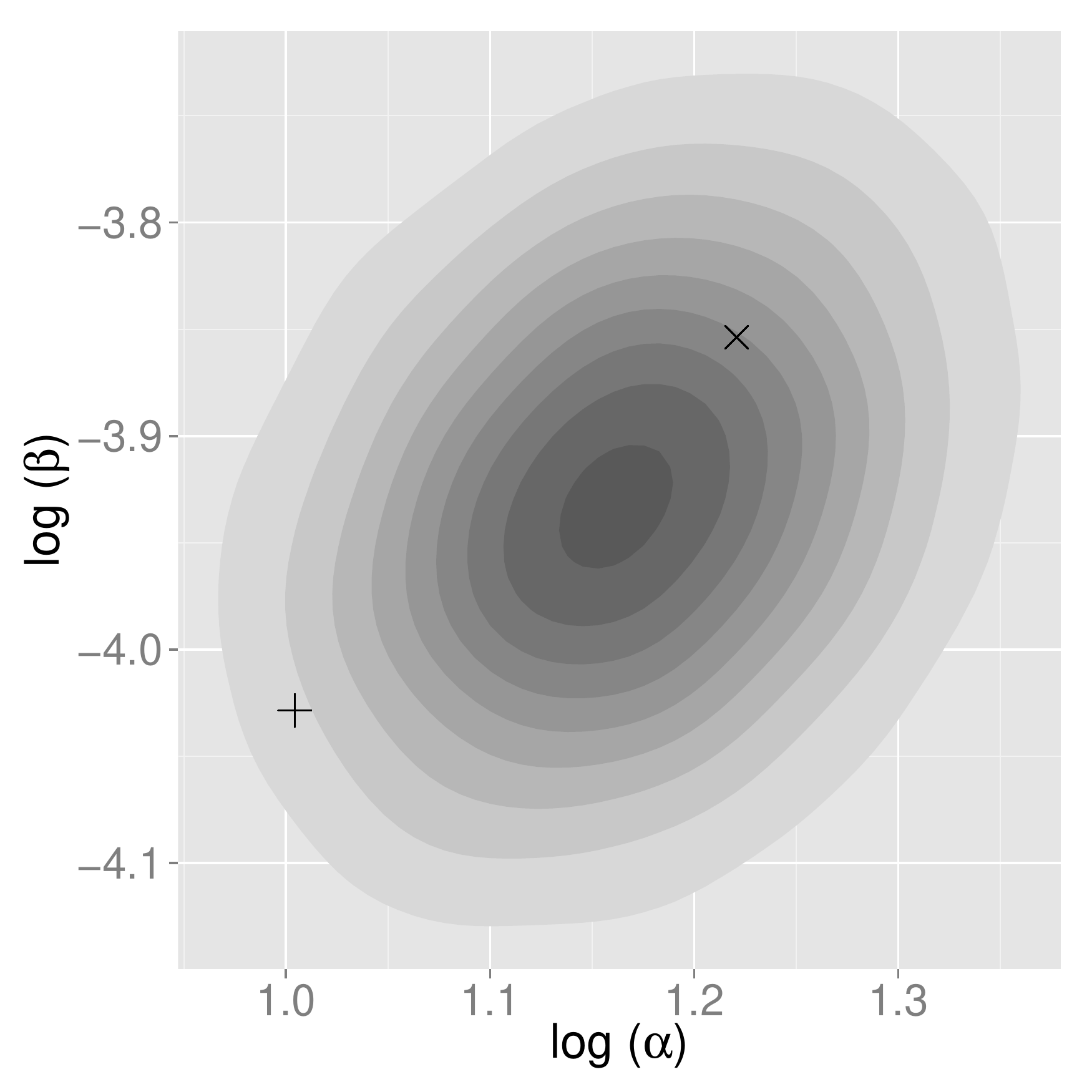}
\label{fig:density}
}
\subfigure[Branching approximation method]{
\includegraphics[width=0.5\textwidth]{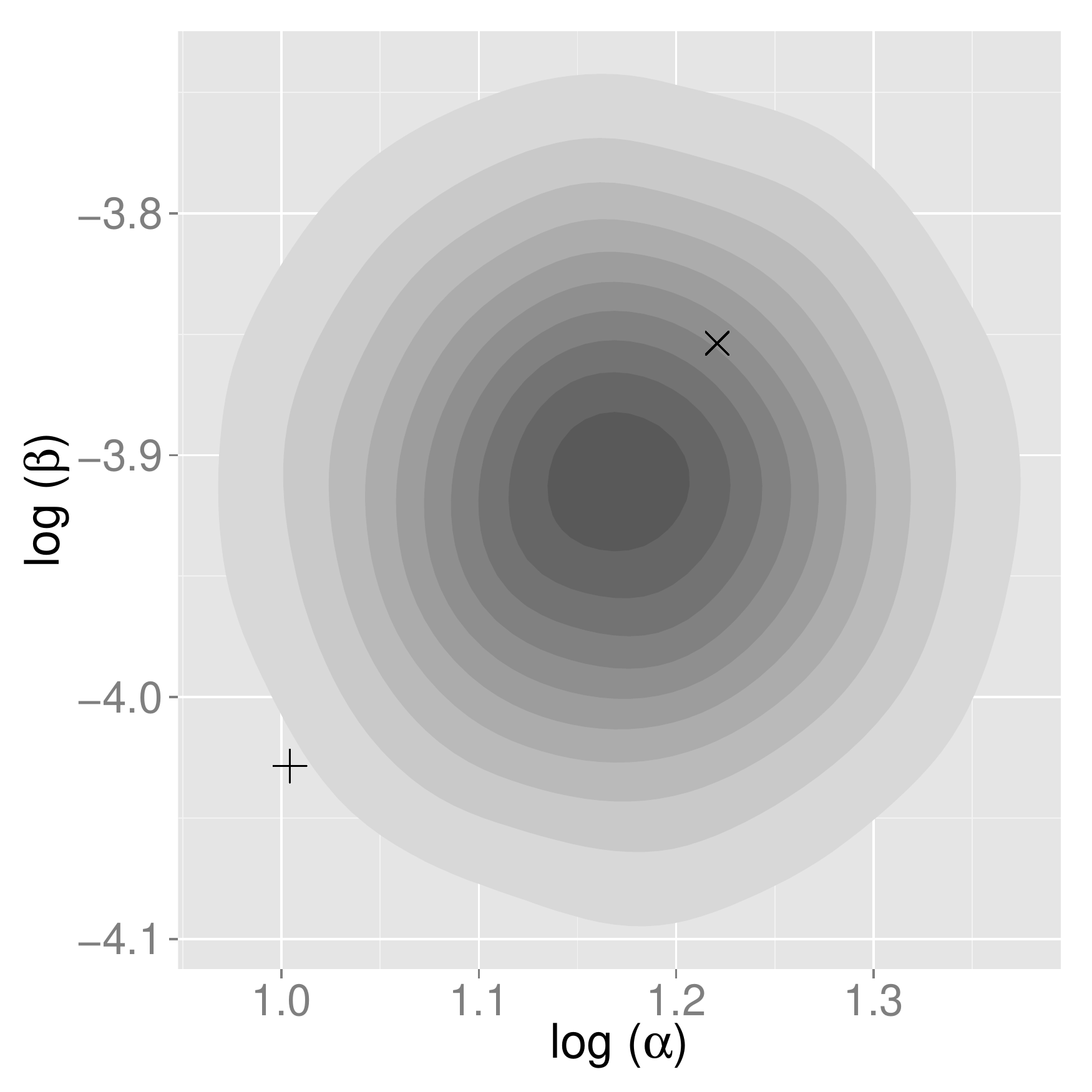}
\label{fig:TransProb_density}
}
\caption{Posterior distributions (log scale) of the death rate $\alpha$ and the infection rate $\beta$ during the plague of Eyam in 1666. The ``$+$" symbol represents the estimate from \citet{brauer2008compartmental} using the deterministic SIR model, and the ``$\times$" symbol represents the Raggett's point estimate.} 
\end{figure}


The posterior distribution of the basic reproduction number $R_0$ from the continued fraction method and from the branching approximation method are similar (Figure \ref{fig:reproduction}). The posterior mean of $R_0$ from the continued fraction method is $1.61$ and from the branching approximation method is $1.62$. The estimate for $R_0$ from \citet{brauer2008compartmental} is $1.7$, from \citet{raggett1982} is $1.63$. These estimates are similar, and in particular the branching approximation estimate is very close to that under the continued fraction method, offering a very efficient way to provide reasonable estimates of quantities such as $R_0$ despite being less accurate than the continued fraction approach.
\begin{figure}[!h]
\begin{center}
\includegraphics[width=0.5\textwidth]{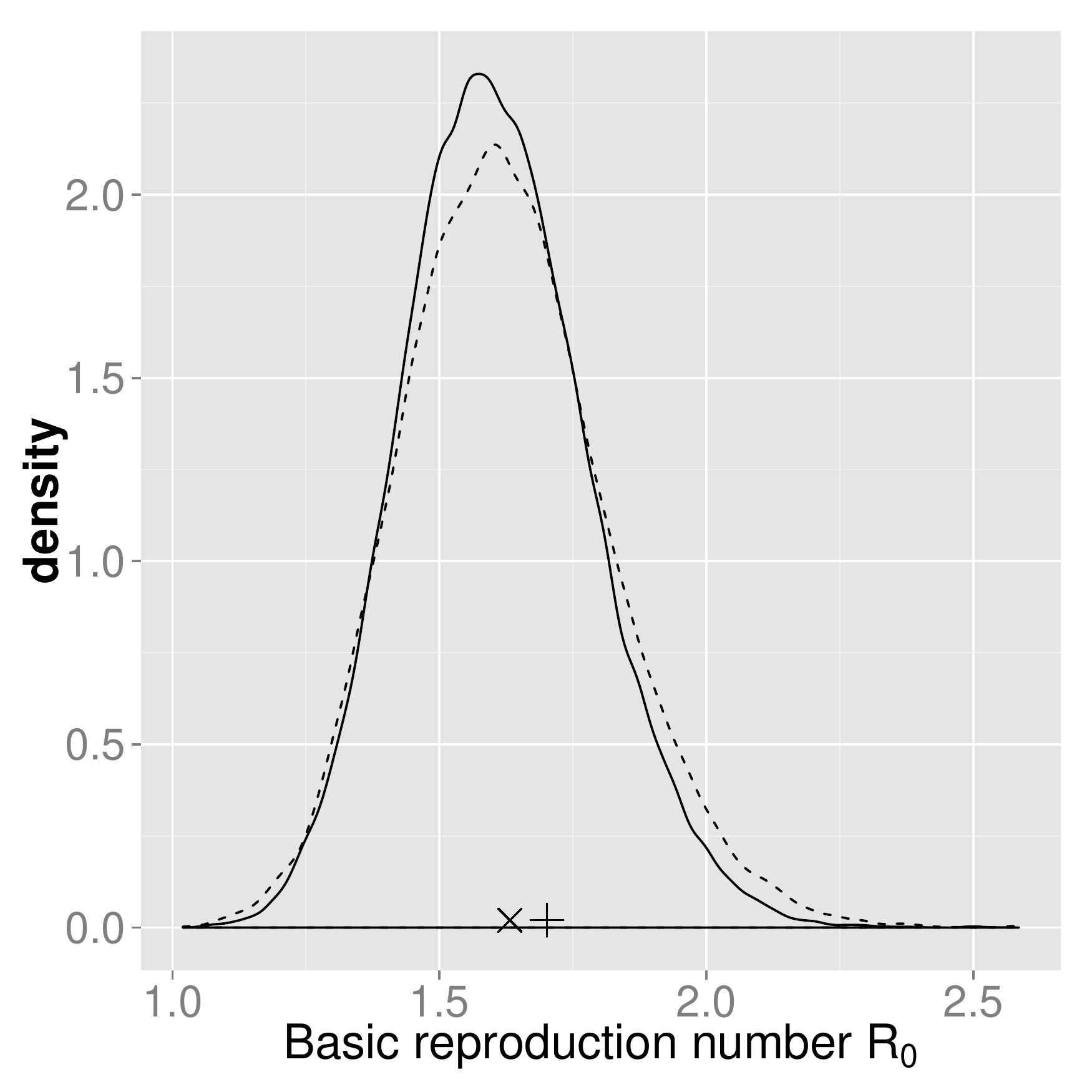}
\end{center}
\caption{Posterior distribution of the basic reproduction number $R_0$ (solid line: continued fraction method, dashed line: branching approximation method). The ``$+$", and the ``$\times$" symbols represent the estimate of $R_0$ from \citet{brauer2008compartmental}, and from \citet{raggett1982} respectively.} 
\label{fig:reproduction}
\end{figure}

From the results, we can see that estimates of $R_0$ from different methods are roughly the same while estimates of $\alpha$ and $\beta$ are different. Although the basic reproduction number $R_0$ is an important quantity in the SIR model, it is not the only parameter driving the dynamic of the epidemic. \citet{correia2014transmission} demonstrated the important of accurately estimating the transmission parameters between compartments of the SIR model for \textit{Salmonella} Typhimurium in pigs.


\newcommand{\MSE}{\text{MSE}}

\section{Discussion}

Likelihood-based inference for bivariate continuous-time Markov processes is usually restricted to very small state spaces due to the computational bottleneck of transition probability calculation. 
In this paper, we provide tools for likelihood-based inference for birth(death)/birth-death processes by developing an efficient method to compute their transition probabilities. 
We provide a complete implementation of the algorithms to compute these transition probabilities in a new \texttt{R} package called \texttt{MultiBD}. 
Our functions employ sophisticated tools
including continued fractions, the modified Lentz method, the method of Abate and Whitt for approximate inverse Laplace transforms, and the Levin acceleration method. Moreover, these methods are naturally amenable to parallelization, and we exploit multicore processing to speed up the algorithm.
We remark that birth(death)/birth-death processes remain a limited subclass of general multivariate birth-death processes. For example, many population biology problems require a full bivariate birth-death process including predator-prey models \citep{hitchcock1986, owen2014} and the SIR model with vital dynamics \citep{earn2008light}. Unfortunately, efficiently computing the transition probabilities of multivariate birth-death processes remains an open problem. Solving this problem will enable numerically stable statistical inference under birth-death processes and will be worth the ``heroic'' effort \citep{renshaw2011}.


\section*{Acknowledgments}

This work was partially supported by the National Institutes of Health (R01 HG006139, R01 AI107034, and U54 GM111274) and the National Science Foundation (IIS 1251151, DMS 1264153, DMS 1606177 ).
We thank Christopher Drovandi, Edwin Michael, and David Denham for access to the \textit{Brugia pahangi} count data.


\numberwithin{equation}{section}
\appendix

\section{Continued fractions}
\label{sec:cf}

In this section, we give some basic definitions and properties related to continued fractions.
\begin{appxdef}
A continued fraction $\phi_0$ is a scalar quantity expressed in
\begin{equation}
\phi_0 =  \cfrac{x_1}{
y_1 + \cfrac{x_2}{
y_2 + \cfrac{x_3}{
y_3 + \cdots
}~,}}
\label{eqn:cf}
\end{equation}
where $\{ x_i \}_{i=1}^\infty$ and $\{ y_i \}_{i=1}^\infty$ are infinite sequences of complex numbers.
\end{appxdef}

\begin{appxdef}
The $n^{\mbox{\tiny th}}$ convergent of $\phi_0$ is
\begin{equation}
\frac{X_n}{Y_n} = \cfrac{x_1}{
y_1 + \cfrac{x_2}{
y_2 + \cfrac{x_3}{
y_3 + \cdots + \cfrac{x_n}{y_n
}~.}}}
\end{equation}
\end{appxdef}

\begin{appxdef}
We define the corresponding sequence $\{\phi_n\}_{n=0}^\infty$ of a continued fraction (\ref{eqn:cf}) by the following recurrence formulae
\begin{equation} 
\begin{aligned}
& \phi_1 = x_1 - y_1\phi_0,~\mbox{and} \\
& \phi_n = x_n \phi_{n-2} - y_n \phi_{n-1}~\mbox{for}~n \geq 2. 
\end{aligned} 
\end{equation} 
\label{def:corrseq}
\end{appxdef}

\citet{murphy1975} provided the following sufficient condition for the convergence of (\ref{eqn:cf}):
\begin{appxlem}
Assume that there exists $N$ such that $\inf_{n >N}|Y_n| > 0$ and $\lim_{n \to \infty} \phi_n = 0$. Then, the continued fraction (\ref{eqn:cf}) is convergent. Moreover,
\begin{equation}
\phi_n = \prod_{i=1}^{n}{x_i} \cfrac{x_{n+1}}{
Y_{n+1} + \cfrac{x_{n+2} Y_n}{
y_{n+2} + \cfrac{x_{n+3}}{
y_{n+3} + \cfrac{x_{n+4}}{
y_{n+4} + \cdots
}~.}}}
\end{equation}
\label{lem:cf}
\end{appxlem}
Now, if we consider a more general recurrence formulae
\begin{equation}
\begin{aligned}
& \phi^{(m)}_1 = - y_1\phi^{(m)}_0 + k_1 1_{\{m = 0\}} \\
& \phi^{(m)}_n = x_n \phi^{(m)}_{n-2} - y_n \phi^{(m)}_{n-1} + k_{m+1} 1_{\{m = n - 1\}}~\text{for}~n \geq 2,
\label{eqn:gen_rec}
\end{aligned} 
\end{equation}
then under the assumption of Lemma \ref{lem:cf}, we have the following lemma:

\begin{appxlem}
The solution for (\ref{eqn:gen_rec}) is
\begin{equation}
\phi^{(m)}_n = \begin{cases} \frac{(-1)^{m-n}k_{m+1}}{\prod_{i=1}^{m+1}{x_i}} Y_n \phi_m, & \mbox{if } n \leq m \\  \frac{k_{m+1}}{\prod_{i=1}^{m+1}{x_i}} Y_m \phi_n, & \mbox{if } n \geq m. \end{cases}
\end{equation}
\label{lem:sys}
\end{appxlem}


\section{Modified Lentz method}
\label{sec:lentz}
Modified Lentz method \citep{lentz1976, thompson1986} is an efficient algorithm to finitely approximate the infinite expression of the continued fraction $\phi_0$ in (\ref{eqn:cf}) to within a prescribed error tolerance. Let $\phi_0^{(n)}$ be the $n^{\mbox{\tiny th}}$ convergence of $\phi_0$, that is $ \phi_0^{(n)} = X_n/Y_n.$ The main idea of Lentz's algorithm lies in using the ratios
\begin{equation}
A_n = \frac{X_n}{X_{n-1}} ~~ \mbox{and} ~~ B_n = \frac{Y_{n-1}}{Y_n}
\end{equation}
to stabilize the computation of $\phi_0^{(n)}$. We can calculate $A_n$, $B_n$, and $\phi_0^{(n)}$ recursively as follows:
\begin{equation}
\begin{aligned}
A_n &= y_n + \frac{x_n}{A_{n-1}} \\
B_n &= \frac{1}{y_n + x_n B_{n-1}} \\ 
\phi_0^{(n)} &= \phi_0^{(n-1)} A_n B_n.
\end{aligned}
\end{equation}
If $\phi_0^{(n)}$ converges to $\phi_0$, then \citet{craviotto1993} show that

\begin{equation}
\left | \phi_0^{(n)} - \phi_0 \right | \leq \frac{|Y_n/Y_{n-1}|}{{\cal I} [Y_n/Y_{n-1}]} \left| \phi_0^{(n)} - \phi_0^{(n-1)} \right | = \frac{|1/B_n|}{{\cal I} [1/B_n]} \left | \phi_0^{(n)} - \phi_0^{(n-1)} \right |,
\end{equation}
where ${\cal I} [Y_n/Y_{n-1}]$ is the imaginary part of $Y_n/Y_{n-1}$ and is assumed to be non-zero. Hence, the Lentz's algorithm terminates when
\begin{equation}
\frac{|1/B_n|}{{\cal I} [1/B_n]} \left | \phi_0^{(n)} - \phi_0^{(n-1)} \right |
\end{equation}
is small enough. However, $A_n$ and $B_n$ can equal zero themselves and cause problem. Hence, \citet{thompson1986} propose a modification for Lentz's algorithm by setting $A_n$ and $B_n$ to a very small number, such as $10^{-16}$, whenever they equal zero. In practice, the algorithm often terminates after small number of iterations. However, in some rare cases where the numerical computation is unstable, it might take too long before the algorithm terminates. So, we set a predefined maximum number of iterations $H$ as a fallback for these cases.


\section{Convergence results of increasing the truncation level}
\label{sec:conv_trunc}

Let $f_{ab}^{(B)}(s)$ be the output of the approximation scheme (\ref{eqn:trunc}) in Theorem \ref{thm:trans}.
In this section, we prove that $f_{ab}^{(B)}(s)$ converges to $f_{ab}(s)$ as $B$ goes to infinity.
To do so, let us consider a truncated birth/birth-death process $\mb{X}^{(B)}(t) = (X^{(B)}_1(t), X^{(B)}_2(t) )$ at truncation level $B$ such that it executes the same process as $\mb{X}(t)$ on the state $\{ a_0, a_0 + 1, a_0 + 2, \ldots \} \times \{ 0, 1, 2, \ldots, B\}$ except that $\lambda^{(2)}_{aB} = 0$.
Define $P_{ab}^{a_0 b_0, (B)}(t)$ be the transition probabilities of $\mb{X}^{(B)}(t)$ and $T_B$ be the hitting time at which $X_2(t)$ first reach state $B+1$.
For any set $S \subset \mathbb{N}^2$, we have
\begin{align*}
\Pr(\mb{X}(t) \in S) &= \Pr(\mb{X}(t) \in S ~|~ T_B > t) \Pr(T>t) + \Pr(\mb{X}(t) \in S ~|~ T_B \leq t) \Pr(T_B \leq t) \\
&= \Pr(\mb{X}^{(B)}(t) \in S) \Pr(T_B > t) + \Pr(\mb{X}(t) \in S ~|~ T_B \leq t) \Pr(T_B \leq t) \\
&= \Pr(\mb{X}^{(B)}(t) \in S) + [\Pr(\mb{X}(t) \in S ~|~ T_B \leq t) - \Pr(\mb{X}^{(B)}(t) \in S)] \Pr(T_B \leq t) 
\end{align*}
Therefore $|\Pr(\mb{X}(t) \in S) - \Pr(\mb{X}^{(B)}(t) \in S)| \leq \Pr(T_B \leq t)$.
Note that $f_{ab}^{(B)}(s)$ is the Laplace transform of $P_{ab}^{a_0 b_0, (B)}(t)$.
Hence
\[
|f_{ab}^{(B)}(s) - f_{ab}(s)| \leq \int_{0}^\infty{|P_{ab}^{a_0 b_0, (B)}(t) - P_{ab}^{a_0 b_0}(t)| e^{-st}dt} \leq \int_{0}^\infty{ Pr(T_B \leq t) e^{-st}dt}
\]
By Dominated convergence theorem and the fact that $\lim_{B \to \infty} Pr(T_B \leq t) = 0$, we deduce that $\lim_{B \to \infty} f_{ab}^{(B)}(s) = f_{ab}(s)$.


\section{Branching SIR approximation}
\label{ap:sir}
Here we derive and solve the Kolmogorov backward equations of the two-type branching process necessary for evaluating the probability generating functions (PGFs) whose coefficients yield transition probabilities. 

\subsection{Deriving the PGF}
Our two-type branching process is represented by a vector $(X_1(t), X_2(t))$ that denotes the numbers of particles of two types at time $t$. Let the quantities $a_1(k,l)$ denote the rates of producing $k$ type 1 particles and $l$ type 2 particles, starting with one type 1 particle, and $a_2(k,l)$ be analogously defined but beginning with one type 2 particle. 
Given a two-type branching process defined by instantaneous rates $a_i(k,l)$, denote the following \textit{pseudo-generating} functions for $i = 1,2$ as
\begin{equation} 
u_i(s_1,s_2) = \sum_k \sum_l a_i(k,l)s_1^k s_2^l .
\end{equation}

We may expand the probability generating functions in the following form:
\begin{align}
\phi_{10}(t, s_1, s_2) &= E (s_1^{X_1(t)} s_2^{X_2(t)} | X_1(0) = 1, X_2(0) = 0)
\\	&= \sum_{k=0}^\infty \sum_{l=0}^\infty P_{1,0}^{kl} (t) s_1^k s_2^l \nonumber
\\ 	&= \sum_{k=0}^\infty \sum_{l=0}^\infty ( \mathbf{1}_{k=1, l = 0} + a_1(k,l) t + o(t) ) \nonumber s_1^k s_2^l
\\ 	&= s_1 + u_1(s_1, s_2) t + o(t) . \nonumber
\end{align}
We have an analogous expression for $\phi_{01}(t, s_1, s_2)$ beginning with one particle of type 2 instead of type 1. For short, we will write $\phi_{10} := \phi_1, \phi_{01} := \phi_2$.
Thus, we have the following relation between the functions $\phi$ and $u$:
\begin{align}
\label{eq:relation}
\frac{d \phi_1}{dt} (t, s_1, s_2) |_{t=0} &= u_1(s_1, s_2) \text{ and} \\
\frac{d \phi_2}{dt} (t, s_1, s_2) |_{t=0} &= u_2(s_1, s_2) . \nonumber
\end{align}

To derive the backwards and forward equations, Chapman-Kolmogorov arguments yield the symmetric relations
\begin{align}
\phi_1(t+h, s_1, s_2) &= \phi_1(t, \phi_1(h, s_1, s_2), \phi_2(h, s_1, s_2))
\\ &= \phi_1(h, \phi_1(t, s_1, s_2), \phi_2(t, s_1, s_2)) . \nonumber
\end{align}
First, we derive the backward equations by expanding around $t$ and applying (\ref{eq:relation}):
\begin{align}
\phi_1(t+h, s_1, s_2) &= \phi_1(t, s_1, s_2) + \frac{d \phi_1}{dh}(t+h, s_1, s_2) | _{h=0} h + o(h) 
\\ &= \phi_1(t, s_1, s_2) + \frac{d \phi_1}{dh}(h, \phi_1(t, s_1, s_2), \phi_2(t, s_1, s_2) |_{h=0} h + o(h) \nonumber
\\ &= \phi_1(t,s_1,s_2) + u_1( \phi_1(t,s_1, s_2) , \phi_2(t, s_1, s_2) h + o(h) ) . \nonumber
\end{align}
Since an analogous argument applies for $\phi_2$, we arrive at the system
\begin{align}
\frac{d}{dt} \phi_1(t, s_1, s_2) &= u_1( \phi_1(t, s_1, s_2), \phi_2(t, s_1, s_2) ) \text{ and} \\
\frac{d}{dt} \phi_2(t, s_1, s_2) &= u_2( \phi_1(t, s_1, s_2), \phi_2(t, s_1, s_2) ) , \nonumber
\end{align}
with initial conditions $\phi_1(0, s_1, s_2) = s_1, \phi_2(0, s_1, s_2) = s_2$.

Recall in our SIR approximation, we use the initial population $X_2(0)$ as a constant that scales the instantaneous rates over any time interval $[t_0, t_1)$. The only nonzero rates specifying this proposed model, in the notation above, are
\begin{equation}\label{eq:SIRrates} a_1(0,1) = \beta X_2(0), \quad \quad  a_1(1,0) = -\beta X_2(0), \quad \quad a_2(0,1) = -\alpha, \quad \quad a_2(0,0) = \alpha. \end{equation}

For simplicity, call $X_2(0) := I_0$, the constant representing the infected population at the beginning of the time interval.
Thus, the corresponding pseudo-generating functions have a simple form:
\begin{align}
u_1(s_1, s_2) = \beta I_0 s_2 - \beta I_0 s_1 \text{ and} \\
u_2(s_1, s_2) = \alpha - \alpha s_2 = \alpha(1 - s_2) . \nonumber
\end{align}
Plugging into the backward equations, we obtain
\begin{align}
\frac{d}{dt} \phi_1(t,s_1,s_2) &= \beta I_0 \big( \phi_2(t, s_1, s_2) - \phi_1(t, s_1, s_2) \big) \text{ and} \\ 
\frac{d}{dt} \phi_2(t,s_1, s_2) &= \alpha - \alpha \phi_2(t,s_1,s_2). \nonumber
\end{align}
The $\phi_2$ differential equation corresponds to a pure death process and is immediately solvable; suppressing the arguments of $\phi_2$ for notational convenience, we obtain
\begin{align}
 \frac{d}{dt} \phi_2 &= \alpha - \alpha \phi_2
 \\ \frac{d}{dt} \phi_2( \frac{1}{1 - \phi_2}) &= \alpha \nonumber
 \\ \ln (1 - \phi_2) &= -\alpha t + C \nonumber
 \\ \phi_2 &= 1 - \exp(-\alpha t + C) . \nonumber
\end{align}
Plugging in $\phi_2(0, s_1, s_2) = s_2$, we obtain $C = \ln(1-s_2)$, and we arrive at 
\begin{equation}\label{phi2}
\phi_2(t, s_1, s_2) = 1 + (s_2 - 1)\exp(-\alpha t) 
\end{equation}

Substituting this solution into the first differential equation and applying the integrating factor method provides
\begin{align}
\phi_1 e^{\beta I_0 t} &= \int \beta I_0 e^{\beta I_0 t} (1 + \frac{s_2-1}{e^{\alpha t}}) \, dt = e^{\beta I_0 t} + \beta I_0 (s_2 -1) \int e^{(\beta I_0 - \alpha )t} \, dt  \\
 &= e^{\beta I_0 t} + \beta I_0 (s_2-1) \frac{e^{(\beta I_0 - \alpha) t}}{\beta I_0 - \alpha} + C . \nonumber
\end{align}
Plugging in the initial condition $\phi_1(0,s_1,s_2) = s_1$ and rearranging yields
\begin{equation}
\phi_1 = 1 + \frac{ \beta I_0 (s_2 - 1)}{\beta I_0 - \alpha} e^{-\alpha t} + e^{-\beta I_0 t} ( s_1 - 1 - \frac{\beta I_0 (s_2 - 1) }{\beta I_0 - \alpha} ) .
\end{equation}

\subsection{Transition probability expressions}
Transition probabilities are related to the PGF via repeated partial differentiation; note that 
\begin{align}
P_{kl}^{mn}(t) &= \frac{1}{k!}\frac{1}{l!}\frac{\partial^k}{\partial s_1^k} \frac{\partial^l}{\partial s_2^l} \phi_{mn}(t, s_1, s_2) \bigg|_{s_1=s_2=0} \\
&= \frac{1}{k!}\frac{1}{l!}\frac{\partial^k}{\partial s_1^k} \frac{\partial^l}{\partial s_2^l} \phi_1^m(t, s_1, s_2) \phi_2^n(t, s_1, s_2) \bigg|_{s_1=s_2=0} \nonumber \\
&=  \frac{\partial^l}{\partial s_2^l} \sum_{i=0}^k {k \choose i} \frac{ \partial^{k-i}}{\partial s_1^ {k-i} } \phi_1^m(t, s_1, s_2) \frac{ \partial^i}{\partial s_1^i} \phi_2^n(t, s_1, s_2)  \bigg|_{s_1=s_2=0} . \nonumber
\end{align}
This expression is generally unwieldy, but notice  $ \frac{ \partial^i}{\partial s_1^i} \phi_2^n(t, s_1, s_2)  \bigg|_{s_1 = 0} = 0 \text{ for all } i > 0$ in our model. Remarkably, this allows us to further simplify and ultimately arrive at closed-form expressions. Continuing, we see
\begin{align}
P_{kl}^{mn}(t) &= \frac{\partial^l}{\partial s_2^l}  \left[ { k \choose 0 } \phi_2^n(t, s_1, s_2)  \frac{ \partial^k}{\partial s_1^k} \phi_1^m(t, s_1, s_2) \right] \bigg|_{s_1=s_2=0}  \\
&=  \frac{\partial^l}{\partial s_2^l} \bigg\{ \phi_2^n(t, s_1, s_2)  \cdot \frac{m!}{(m-k)!} e^{-k \beta I_0 t} \bigg[ 1 + \frac{ \beta I_0 (s_2 -1)}{\beta I_0 - \alpha } e^{-\alpha t} \nonumber \\ & \hspace{2.7in} - e^{-\beta I_0 t} \bigg( 1 + \frac{ \beta I_0 (s_2 - 1)}{ \beta I_0 - \alpha } \bigg) \bigg] ^{m-k} \bigg\} \bigg|_{s_1=s_2=0}  \nonumber \\
&:=  \frac{\partial^l}{\partial s_2^l}  \left[ \phi_2^n(t, s_1, s_2) \cdot  h(t, s_1, s_2) \right] \bigg|_{s_1=s_2=0} \nonumber \\
&= \sum_{i=0}^l {l \choose i} \frac{ \partial^{l-i}}{\partial s_2^ {l-i}} h(t, s_1, s_2)  \frac{\partial^i}{\partial s_2^i}  \phi_2^n(t, s_1, s_2) \nonumber \\
&:= \sum_{i=0}^l {l \choose i} A(l-i) B(i) . \nonumber
\end{align} 
From here, it is straightforward to take partial derivatives of $h(t, s_1, s_2)$ and our closed-form expression of $\phi_2^n(t, s_1, s_2)$ to arrive at Conditions (\ref{eq:conditionB}) and (\ref{eq:conditionA}).
A heatmap visualization of the difference between transition probabilities under the branching approximation and those computed using the continued fraction method for the SIR model is included below.
\begin{figure}[H]
\centering
    \includegraphics[scale=.7]{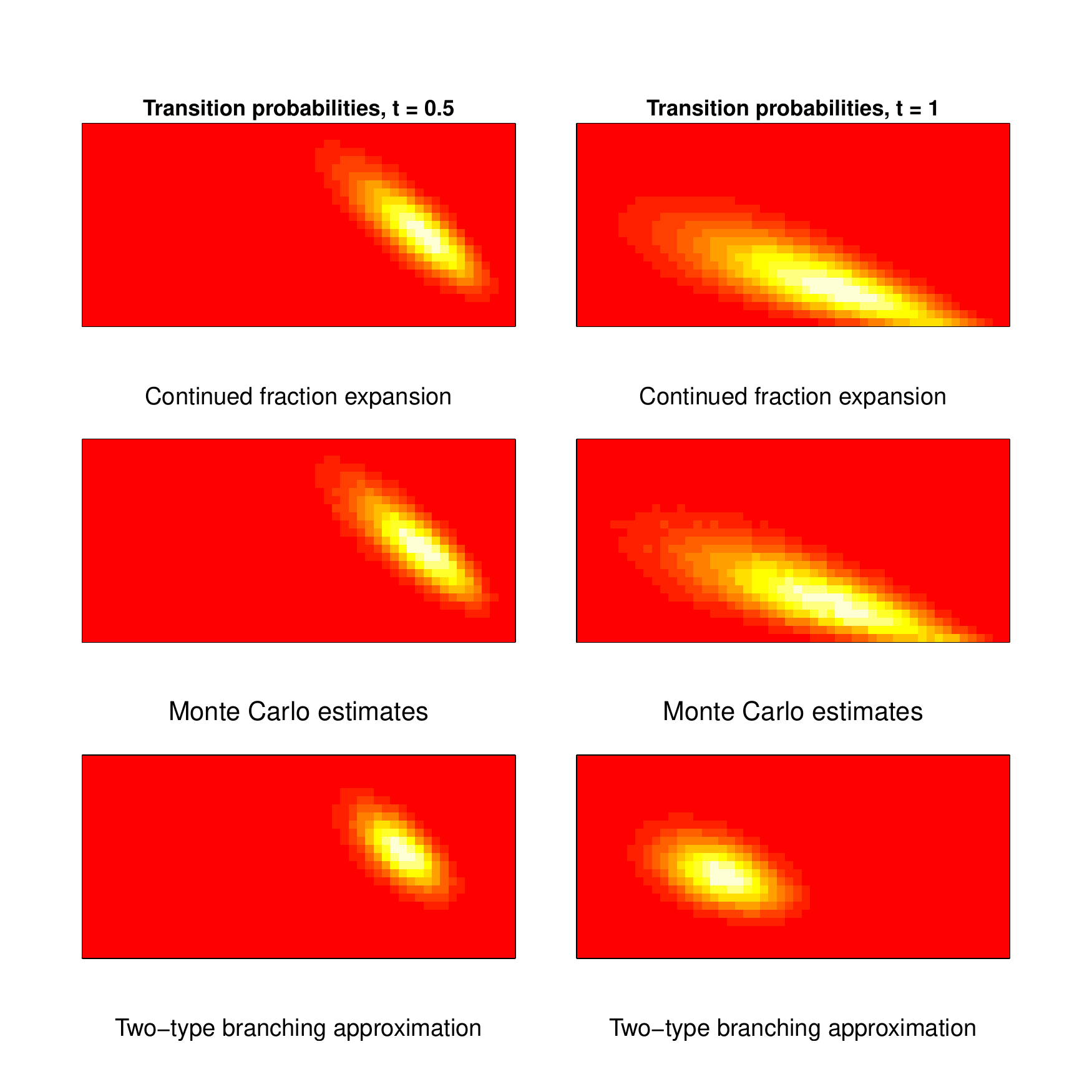} 
\caption{ Heatmap visualizations of transition probabilities near the region of support across methods for $t=0.5, 1$. We see that the branching approximation is noticeably different from the Monte Carlo ground truth when we increase $t$ to $1$, while the continued fraction approach remains accurate. }    
\label{fig:heatmap}
\end{figure}


\bibliographystyle{chicago}
\bibliography{ms}
\end{document}